\documentclass[journal, twoside, web]{Classes/IEEEtran}

\usepackage{Packages/generic}

\usepackage{Packages/myFormat}

\def\BibTeX{{\rm B\kern-.05em{\sc i\kern-.025em b}\kern-.08emT\kern-.1667em\lower.7ex\hbox{E}\kern-.125emX}}%
\title{\vspace{10pt}Matrix-Scheduling of QSR-Dissipative Systems}%
\author{%
    Sepehr Moalemi, and James Richard Forbes
    \thanks{%
        Sepehr Moalemi {\tt\small sepehr.moalemi@mail.mcgill.ca} and James Richard Forbes {\tt\small james.richard.forbes@mcgill.ca} are with the Department of Mechanical Engineering, McGill University, 817 Sherbrooke St. W., Montreal, QC H3A 0C3, Canada.%
    }%
}
\begin{document}
    \maketitle

    \fontdimen16\textfont2=\fontdimen17\textfont2
    \fontdimen13\textfont2=5pt

    \begin{abstract}
    This paper considers gain-scheduling of QSR-dissipative subsystems using scheduling matrices. The corresponding QSR-dissipative properties of the overall matrix-gain-scheduled system, which depends on the QSR properties of the subsystems scheduled, are explicitly derived. The use of scheduling matrices is a generalization of the scalar scheduling signals used in the literature, and allows for greater design freedom when scheduling systems, such as in the case of gain-scheduled control. Furthermore, this work extends the existing gain-scheduling results to a broader class of QSR-dissipative systems. The matrix-scheduling of important special cases, such as passive, input strictly passive, output strictly passive, finite $\mathcal{L}_2$ gain, very strictly passive, and conic systems are presented. The proposed gain-scheduling architecture is used in the context of controlling a planar three-link robot subject to model uncertainty. A novel control synthesis technique is used to design QSR-dissipative subcontrollers that are gain-scheduled using scheduling matrices. Numerical simulation results highlight the greater design freedom of scheduling matrices, leading to improved performance.
\end{abstract}
    \section{Introduction}
\IEEEPARstart{D}{issipative} systems theory~\cite{J_C_Willems} is a general and broadly applicable systems theoretic framework that characterize a dynamical system by its input-output behavior, and it strongly relates to Lyapunov and \(\mathcal{L}_2\) stability theories~\cite{khalil}. As discussed in~\cite{van_der_schaft}, the input-output dissipativity theory in~\mbox{\cite{Brogliato, Zames_1}} is an extension of classical Lyapunov based dissipativity theory in~\mbox{\cite{J_C_Willems, Hill_Moylan_QSR}}, where \emph{a priori} information between input and output variables is not distinguished. Through the introduction of a storage function, these two approaches were unified in~\cite{J_L_Willems}. One such class of storage functions are quadratic supply rate (QSR) functions, which are used to analyze the stability of interconnected systems~\cite{QSR_Stability, Hill_Moylan_QSR}. A special case of QSR-dissipativity for square systems, whose inputs and outputs have the same dimension, is passivity. Over the years, passivity and dissipativity have seen a wide range of applications, including exponential stabilization of dynamical
systems~\cite{Fradkov_Hill}, analysis of event-triggered networked control systems~\cite{Panos_Rahnama}, synchronization of neural networks~\cite{Zhang_Liu}, and control of Euler\textendash Lagrange systems~\cite{Ortega_Nicklasson}, flexible multi-link manipulators~\cite{Damaren_passivity_multilink}, and delayed teleoperations~\cite{Panzirsch}.
\begin{figure}[t]
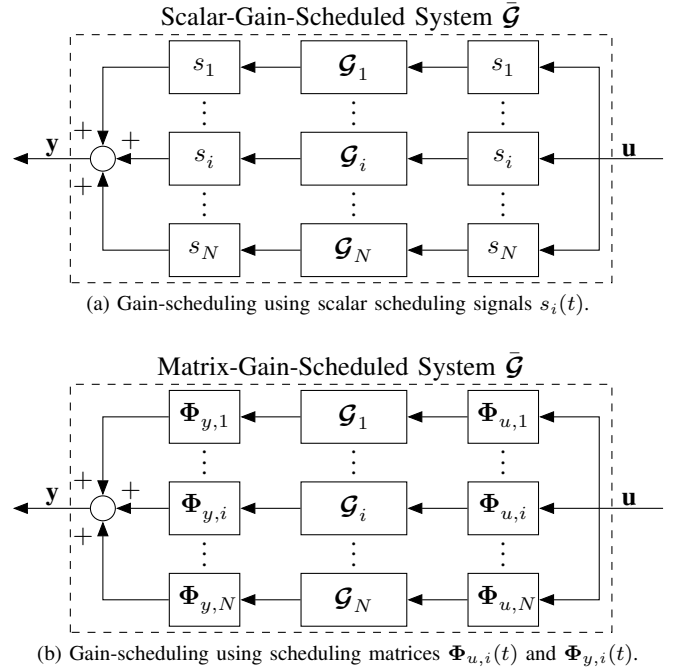

    \centering
    \subfloat[Gain-scheduling using scalar scheduling signals \(s_{{\it i}}(t)\).]{%
        \includegraphics{Figures/scalar_scheduling_intro_simplified.tikz}%
        \label{fig:scalar_scheduling}%
    }%
    \qquad
    \subfloat[Gain-scheduling using scheduling matrices \(\mbs{\Phi}_{u, {\it i}}(t)\) and \(\mbs{\Phi}_{y, {\it i}}(t)\).]{%
        \includegraphics{Figures/matrix_scheduling_intro_simplified.tikz}%
        \label{fig:matrix_scheduling}%
    }%
    \caption{Gain-scheduling of \(N\) subsystems using two different gain-scheduling architectures. The scalar scheduling signals \(s_{{\it i}}(t)\) and the scheduling matrices \(\mbs{\Phi}_{u, {\it i}}(t)\) and \(\mbs{\Phi}_{y, {\it i}}(t)\) are designed to interpolate between the subsystems to achieve acceptable performance.}%
    \label{fig:combined}
    \vspace{-5pt}
\end{figure}

When controlling a nonlinear system, it is possible to design linear controllers using the linearized model of a system about an operating point. However, the linearized model of the plant may not adequately capture the behavior of the nonlinear system over a wide range of operating conditions. Therefore, a single linear controller designed using a single linearized plant may not realize adequate closed-loop performance, or even closed-loop stability, across a wide range of operating conditions. Gain-scheduled control is a nonlinear control technique where a set of linear subcontrollers are designed about multiple linearization points across the operating range of a nonlinear system to be controlled. The linear subcontrollers are then interpolated or scheduled to form a gain-scheduled controller that realizes acceptable performance, ideally with closed-loop stability guarantees as well.

Recently, the stability of gain-scheduled controllers has been studied through the lens of passivity~\cite{Damaren_passive_map, Forbes_Damaren, Forbes_Damaren_Stability}, conicity~\mbox{\cite{Forbes_thesis, Ryan_conic}}, and dissipativity~\cite{QSR}. The gain-scheduling architecture in \mbox{\Cref{fig:scalar_scheduling}} was first introduced in~\cite{Damaren_passive_map} where it was shown that the gain-scheduling of input strictly passive (ISP) subcontrollers, as per \Cref{fig:scalar_scheduling}, results in an overall ISP gain-scheduled controller. Using the same gain-scheduling architecture, the results of~\cite{Damaren_passive_map} were further extended to very strictly passive (VSP), finite \(\mathcal{L}_2\) gain, and conic systems in~\cite{Forbes_Damaren, Forbes_Damaren_Stability, Forbes_thesis, Ryan_conic}. By leveraging the passivity and conic sector theorems, respectively, the gain-scheduling results of~\cite{Damaren_passive_map, Forbes_Damaren, Forbes_Damaren_Stability, Forbes_thesis, Ryan_conic} are used to guarantee \(\mathcal{L}_2\) stability of robotic and aerospace systems. More recently, the gain-scheduling of QSR-dissipative subcontrollers, as per \Cref{fig:scalar_scheduling}, was considered in~\cite{QSR}. The results of~\cite{QSR} are of particular interest, as subsystems to be gain-scheduled are not required to have identical QSR-dissipative properties, nor be square. Alternate scalar-gain-scheduling architectures, accounting for actuator saturation~\cite{Walsh_saturation} and affine parameter dependence~\cite{Alex_LPV}, have also been studied within the context of passivity-based control.

The gain-scheduling architectures presented in~\cite{Damaren_passive_map, Forbes_Damaren, Forbes_Damaren_Stability, Forbes_thesis, Walsh_saturation, Alex_LPV, QSR, Ryan_conic} all involve \emph{scalar scheduling signals}. As shown in \Cref{fig:scalar_scheduling}, scalar scheduling signals, denoted \(s_i\), scale the entire input-output map of the subsystems based on the same scheduling signal and its parameters. When controlling multiple-input multiple-output (MIMO) systems, each control variable may require a different scheduling profile to switch between subcontrollers. In~\cite{Damaren_GS_Row_Signals}, a row of scalar scheduling signals is used to introduce a notion of extended passive systems. However, the scheduling of special classes of passive systems, such as ISP, output strictly passive (OSP), finite \(\mathcal{L}_2\) gain, VSP, and conic systems is not considered in~\cite{Damaren_GS_Row_Signals}. More recently in~\cite{moalemi_forbes_ccta}, the gain-scheduling of VSP subcontrollers using \emph{scheduling matrices} was shown to result in an overall gain-scheduled VSP controller. The extension of scheduling rows to scheduling matrices allows for coupling between the scheduling signals of different control variables through off-diagonal terms. Nevertheless, the gain-scheduling architecture presented in~\cite{moalemi_forbes_ccta} is limited to VSP systems and cannot be used for the scheduling of a broader class of QSR-dissipative systems. This paper presents the generalized matrix-gain-scheduling architecture in \Cref{fig:matrix_scheduling} that can be used to schedule a broader class of QSR-dissipative systems. Specifically, the contributions of this paper are as follows:
\begin{enumerate}
    \item{%
        Presenting a generalized matrix-gain-scheduling architecture along with a detailed discussion on the \emph{design} and \emph{construction} of scheduling matrices.
    }%
    \item{%
        Extending the gain-scheduling results of~\cite{QSR, moalemi_forbes_ccta} by considering the scheduling of a broader class of QSR-dissipative systems using the proposed matrix-gain-scheduling architecture. The scheduling of important special cases of QSR-dissipativity, such as passive, ISP, OSP, finite $\mathcal{L}_2$ gain, VSP, and conic systems are shown to result in an overall gain-scheduled system of the same special case.
    }%
    \item{%
        Presenting a linear matrix inequality (LMI)-based QSR-dissipative control synthesis technique that renders the feedback system asymptotically stable. This proposed technique is inspired by the LMI-based QSR-dissipative control synthesis technique in~\cite{QSR} and the linear-quadratic-Gaussian (LQG) design in~\cite{Benhabib}.
    }%
\end{enumerate}

The remainder of this paper is as follows. Notation and preliminaries are presented in \Cref{sec:preliminaries}. The gain-scheduling architecture, including properties and construction of scheduling matrices, is introduced in \Cref{sec:gain_scheduling_architecture}. Two theorems classifying the QSR-dissipativity of a matrix-scheduled system composed of $N$ QSR-dissipative subsystems are discussed in \Cref{sec:main_results}. In \Cref{Discussion}, the special cases of QSR-dissipativity are used to compare the results of the proposed theorems to existing results in the literature. An application of the proposed gain-scheduling architecture, complete with the proposed LMI-based QSR-dissipative control synthesis technique is presented in \Cref{sec:simulation}, followed by closing remarks in \Cref{sec:closing_remarks}.
    \section{Preliminaries} \label{sec:preliminaries}
\subsection{Notation}
    Scalars are denoted \(\alpha \in \mathbb{R}\), matrices are denoted \(\mbf{A} \in \mathbb{R}^{m \times n}\), and column matrices are denoted \(\mbf{v} \in \mathbb{R}^{n}\). The identity and zero matrices are \(\eye\) and \(\mbf{0} \), respectively. The set \(\mathbb{S}^n\) denotes the set of real symmetric matrices of size \(n \times n\). The set of symmetric negative semidefinite matrices is denoted as \(\mathbb{S}^{n}_{-}\), the set of symmetric negative definite matrices is denoted as \(\mathbb{S}^{n}_{--}\), the set of symmetric positive semidefinite matrices is denoted as \(\mathbb{S}^{n}_{+}\), and the set of symmetric positive definite matrices is denoted as \(\mathbb{S}^{n}_{++}\)~\cite{boyd2004convex}. The maximum eigenvalue and singular value of \(\mbf{A}\) are denoted as \(\lambda_{\max}(\mbf{A})\) and \(\sigma_{\max}(\mbf{A})\), respectively. The notation $\diag(\cdot)$ denotes a block diagonal matrix of its arguments. Operators are denoted by \(\bm{\mathcal{G}}\), and index sets are denoted by \(\mathcal{N}\).
\subsection{Definitions}
    \begin{definition}[Induced Matrix Norm {\mycite{Zhou_Robust_Control}}]
        Let \(\mbf{A} \in \mathbb{R}^{m \times n}\), then the matrix norm induced by a vector \(p\)-norm is defined as \(\mleft\| \mbf{A} \mright\|_{p} = \sup_{\mbf{x} \neq \mbf{0}} \mleft\| \mbf{A}\mbf{x} \mright\|_{p}/\mleft\| \mbf{x} \mright\|_{p}\). The special case of \(p = 2\) can be written as \(\mleft\| \mbf{A} \mright\|_{2} = \sqrt{\lambda_{\max}(\mbf{A}^{\trans}\mbf{A})} = \sigma_{\max}(\mbf{A})\).
    \end{definition}
    \vspace{3pt}
    \begin{definition}[Truncated Signal {\mycite{Marquez}}]
        Given a signal \(\mbf{u} : \mathbb{R}_{\geq 0} \to \mathbb{R}^{n} \), the truncated signal, \(\mbf{u}_T\), is defined as \(\mbf{u}_T(t) = \mbf{u}(t)\) for \(0 \leq t \leq T \) and \(\mbf{u}_T(t) = \mbf{0}\) for \(t > T \in \mathbb{R}_{\geq 0}\).
    \end{definition}
    \vspace{3pt}
    \begin{definition}[Truncated Inner Product {\mycite{Marquez}}]
        Given signals \(\mbf{u},\,\mbf{y} : \mathbb{R}_{\geq 0} \to \mathbb{R}^{n}\), the truncated  inner product is defined as \(\left\langle \mbf{u}, \mbf{y}\right\rangle_T = \left\langle\mbf{u}_T, \mbf{y}_T\right\rangle = \int_{0}^{T} \mbf{u}^{\trans}(t) \mbf{y}(t)\,\dt,\;\forall T \in \mathbb{R}_{\geq 0} \).
    \end{definition}
    \vspace{3pt}
    \begin{definition}[$\mathcal{L}_p$ Signal Spaces {\mycite{Marquez}}] \label{def:Lp_spaces}
        Given a piecewise continuous signal \(\mbf{u} : \mathbb{R}_{\geq 0} \to \mathbb{R}^{n} \), \(\mbf{u} \in \mathcal{L}_{2e}\) if \(\mleft\| \mbf{u} \mright\|^{2}_{2T} = \left\langle \mbf{u}, \mbf{u}\right\rangle_T < \infty\), \(\forall T \in \mathbb{R}_{\geq 0}\). Additionally, \(\mbf{u} \in \mathcal{L}_{\infty}\) if \(\mleft\| \mbf{u} \mright\|_{\infty} = \sup_{t \in \mathbb{R}_{\geq 0}} \max_{i=1, \ldots, n} \lvert u_i(t) \rvert < \infty\). 
    \end{definition}
    \vspace{3pt}
    \begin{definition}[QSR-Dissipativity~\cite{QSR, Hill_Moylan_QSR}] \label{def:QSR-Dissipativity}
        Consider a causal continuous-time system, \(\bm{\mathcal{G}} : \mathcal{L}_{2e} \rightarrow \mathcal{L}_{2e} \), defined by
        \begin{align*}
            \dot{\mbf{x}}(t) &= \mbf{f}(\mbf{x}(t), \mbf{u}(t)), 
            &
            \mbf{y}(t) &= \mbf{h}(\mbf{x}(t), \mbf{u}(t)),
        \end{align*}
        where \(\mbf{x} \in \mathbb{R}^{n_x}\), \(\mbf{y} \in \mathbb{R}^{n_y}\), and \(\mbf{u} \in \mathbb{R}^{n_u}\). The system \(\bm{\mathcal{G}}\) is QSR-dissipative with matrices \(\mbf{Q} \in \mathbb{S}^{n_y}\), \(\mbf{S} \in \mathbb{R}^{n_y \times n_u}\), and \(\mbf{R} \in \mathbb{S}^{n_u}\) if for all \(\mbf{u} \in \mathcal{L}_{2e}\) and \(T \in \mathbb{R}_{\geq 0}\) there exists a storage function \(V : \mathbb{R}^{n_x} \to \mathbb{R}_{\geq 0}\) such that
        \begin{equation*}\label{eqn:QSR-dissipativity}
            \left\langle \mbf{y}, \mbf{Q} \mbf{y} \right\rangle_T
            + 
            2 \left\langle \mbf{y}, \mbf{S} \mbf{u} \right\rangle_T
            + 
            \left\langle \mbf{u}, \mbf{R} \mbf{u} \right\rangle_T
            \geq 
            V(\mbf{x}(T)) - V(\mbf{x}(0)).
        \end{equation*}
        The special cases of QSR-dissipativity are
        \begin{samepage}
            \begin{itemize}
                \item{%
                    \emph{passive} if \(\mbf{Q} = \mbf{0}\), \(\mbf{S} = \frac{1}{2}\mbf{1}\), and \(\mbf{R} = \mbf{0}\),
                }%
                \item{%
                    \emph{input strictly passive (ISP)} if \(\mbf{Q} = \mbf{0}\), \(\mbf{S} = \frac{1}{2}\mbf{1}\), and \(\mbf{R} = -\delta \mbf{1}\), where \(\delta \in \mathbb{R}_{>0}\),
                }%
                \item{%
                    \emph{output strictly passive (OSP)} if \(\mbf{Q} = -\varepsilon\mbf{1}\), \(\mbf{S} = \frac{1}{2}\mbf{1}\), and \(\mbf{R} = \mbf{0} \), where \(\varepsilon \in \mathbb{R}_{>0}\),
                }%
                \item{%
                    \emph{finite \(\mathcal{L}_{2}\) gain} if \(\mbf{Q} = -\mbf{1}\), \(\mbf{S} = \mbf{0}\), and \(\mbf{R} = \gamma^2 \mbf{1}\), where \(\gamma \in \mathbb{R}_{>0}\),
                }%
                \item{%
                    \emph{very strictly passive (VSP)} if \(\mbf{Q} = -\varepsilon\mbf{1}\), \(\mbf{S} = \frac{1}{2}\mbf{1}\), and \(\mbf{R} = -\delta \mbf{1}\), where \(\delta , \varepsilon \in \mathbb{R}_{>0}\), and
                }%
                \item{%
                    \emph{conic} if \(\mbf{Q} = -\mbf{1}\), \(\mbf{S} = \frac{a + b}{2} \mbf{1} = c\mbf{1}\), and \(\mbf{R} = -ab\mbf{1} = (r^2 - c^2)\mbf{1}\), where \(a, b \in \mathbb{R}\) represent lower and upper conic bounds, while \(c \in \mathbb{R}\) and \(r \in \mathbb{R}_{>0}\) represent the center and radius of the conic sector, respectively.
                }%
            \end{itemize}
        \end{samepage}
    \end{definition}
    \section{Matrix-Gain-Scheduling Architecture} \label{sec:gain_scheduling_architecture}
\begin{figure}[t]
    \vspace{2pt}
    \centering
    \includegraphics{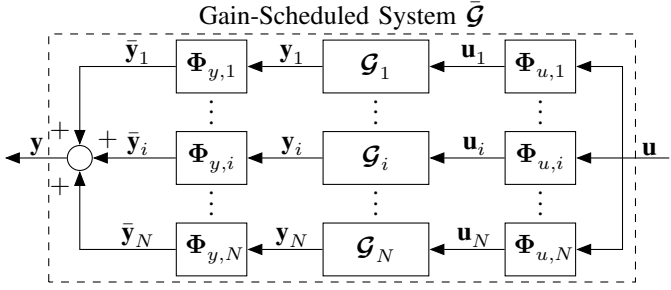}
    \vspace{-10pt}
    \caption{Gain-scheduled system $\bar{\bm{\mathcal{G}}}$, composed of $N$ parallel QSR-dissipative subsystems. The input and output of each subsystem are scheduled as per \cref{eq:GS_QSR_io} via the matrix multiplication between the scheduling matrices \(\mbs{\Phi}_{u, {\it i}}(t)\) and \(\mbs{\Phi}_{y, {\it i}}(t)\), and their corresponding signals \(\mbf{u}(t)\) and \(\mbf{y}_{{\it i}}(t)\), respectively.}
    \label{fig:GS_G}
    \vspace{-5pt}
\end{figure}
This section presents the matrix-gain-scheduling architecture in \Cref{fig:GS_G}. The scheduling functions to be designed take the form of scheduling matrices, allowing for greater design freedom compared to the scalar scheduling signals in~\cite{Damaren_passive_map, Forbes_Damaren, Forbes_Damaren_Stability, Forbes_thesis, Walsh_saturation, Alex_LPV, QSR, Ryan_conic}.%

Consider the gain-scheduled system, $\bar{\bm{\mathcal{G}}}$, in \Cref{fig:GS_G}, composed of $N$ parallel QSR-dissipative subsystems, $\bm{\mathcal{G}}_1, \ldots, \bm{\mathcal{G}}_N$. The subsystems could be linear or nonlinear, and their input-output maps are given by \(\mbf{y}_i(t) = \bm{\mathcal{G}}_i \mleft( \mbf{u}_i(t)\mright)\) for \mbox{\(i \in \mathcal{N} = \{1, \ldots, N \}\)}. The input and output of each subsystem are scheduled in a multiplicative manner such that
\begin{subequations} \label{eq:GS_QSR_io}
    \begin{align} 
        \mbf{u}_i(t) &= \mbs{\Phi}_{u, i}(t) \mbf{u}(t), \label{eq:GS_io_u_i}
        \\
        \bar{\mbf{y}}_i(t) &= \mbs{\Phi}_{y, i} (t) \mbf{y}_{i}^{}(t), \label{eq:GS_io_y_i}
    \end{align}
\end{subequations}
where $\mbs{\Phi}_{u, i}^{} \in \mathbb{R}^{n_u \times n_u}$ and $\mbs{\Phi}_{y, i}^{} \in \mathbb{R}^{n_y \times n_y}$ are the scheduling matrices, $\mbf{u}_i, \mbf{u} \in \mathbb{R}^{n_u}$, and $\mbf{y}_i, \mbf{y} \in \mathbb{R}^{n_y}$, for all $i \in \mathcal{N}$. The gain-scheduled system, $\bar{\bm{\mathcal{G}}}$, input-output map can be written in terms of the individual subsystems inputs, outputs, and scheduling matrices as
\begin{align} \label{eq:GS_io_y_c}
    \mbf{y}(t) = \sum_{i \in \mathcal{N}} \bar{\mbf{y}}_i(t) 
    &= \sum_{i \in \mathcal{N}}\mbs{\Phi}_{y, i} (t) \mbf{y}_{i}^{}(t)\nonumber \\%
    &= \sum_{i \in \mathcal{N}}\mbs{\Phi}_{y, i} (t)\bm{\mathcal{G}}_i \mleft( \mbf{u}_i(t)\mright)\nonumber \\%
    &= \sum_{i \in \mathcal{N}} \mbs{\Phi}_{y, i} (t) \bm{\mathcal{G}}_{i}\mleft(\mbs{\Phi}_{u, i}(t) \mbf{u}(t)\mright).
\end{align}

Based on the QSR-dissipativity property in \Cref{def:QSR-Dissipativity}, the input and output of each subsystem must satisfy
\begin{equation}
        \hspace{-5pt}
        \left\langle \mbf{y}_i, \mbf{Q}_i \mbf{y}_i \right\rangle_T
        + 
        2 \left\langle \mbf{y}_i, \mbf{S}_i \mbf{u}_i \right\rangle_T
        + 
        \left\langle \mbf{u}_i, \mbf{R}_i \mbf{u}_i \right\rangle_T
        \geq 
        \tilde{V}_i, \, \forall T \in \mathbb{R}_{\geq 0}, \label{eq:QSR-Dissipativity_i}
\end{equation} 
with \(\tilde{V}_i = V_i(\mbf{x}(T)) - V_i(\mbf{x}(0))\), \(\mbf{Q} \in \mathbb{S}^{n_y}\), \(\mbf{S} \in \mathbb{R}^{n_y \times n_u}\), and \(\mbf{R} \in \mathbb{S}^{n_u}\) for all \(i \in \mathcal{N}\) and \(\mbf{u}_{i} \in \mathcal{L}_{2e}\).

The objective is to \emph{design} the scheduling matrices $\mbs{\Phi}_{u, i}(t)$ and $\mbs{\Phi}_{y, i}(t)$ for $i \in \mathcal{N}$, such that the gain-scheduling of the $N$ QSR-dissipative subsystems results in an overall gain-scheduled QSR-dissipative system. Furthermore, if all subsystems belong to the same special case of QSR-dissipativity as defined in \Cref{def:QSR-Dissipativity}, the scheduling matrix design should result in the overall gain-scheduled system belonging to that same special case.

\subsection{Scheduling Matrix Properties} \label{subsec:scheduling_matrix_properties}
Consider the set of scheduling matrices \(\mbs{\Phi}_{j, i}(t) \in \mathbb{R}^{n_j \times n_j}\) for \(j \in \mleft\{u, y\mright\}\) and \(i \in \mathcal{N}\). With abuse of set notation, denote the time dependent set, $\mathcal{F}_{j}(t)$, as the index set of all the respective full rank scheduling matrices at time $t \in \mathbb{R}_{\geq 0}$. That is, for \(t \in \mathbb{R}_{\geq 0}\) and \(j \in \mleft\{u, y\mright\}\), 
\begin{equation} \label{eqn:Fj(t)}
    \mathcal{F}_j(t) = \mleft\{\,i \in \mathcal{N} \mid \rank \mleft(\mbs{\Phi}_{j, i}(t) \mright) = n_j \,\mright\}.
\end{equation}

\begin{definition}[Active Scheduling Matrices]
    Given a gain-scheduled system of the type shown in \Cref{fig:GS_G} with scheduling matrices $\mbs{\Phi}_{j, i}(t)\in \mathbb{R}^{n_j \times n_j}$ for \(j \in \{u, y\}\) and $i \in \mathcal{N}$, the scheduling matrices are said to be
    \begin{itemize}
        \item \emph{active} if at all times, there exists at least one nonzero scheduling matrix, meaning $\forall t \in \mathbb{R}_{\geq 0}$, \(\exists i \in \mathcal{N}\) such that \(\mbs{\Phi}_{j, i}(t) \neq \mbf{0}\) and
        \item \emph{strongly active} if at all times, there exists at least one full rank scheduling matrix, meaning $\forall t \in \mathbb{R}_{\geq 0}$, \(\exists i \in \mathcal{N}\) such that \(\rank \mleft(\mbs{\Phi}_{j, i}(t) \mright) = n_j\). Using set notation, this can be written as $\forall t \in \mathbb{R}_{\geq 0}$, \(\mathcal{F}_j(t) \neq \varnothing\).
    \end{itemize}
\end{definition}
When considering the special cases of QSR-dissipativity in \Cref{Discussion}, the active and strongly active cases of scheduling matrices are invoked to correctly quantify the overall QSR-dissipativity of the gain-scheduled system. 

For the remainder of this paper the input and output scheduling matrices are assumed to be bounded in the sense that
\begin{align} \label{eqn:scheduling_matrix_boundedness}
    \sup_{t \in \mathbb{R}_{\geq 0}}\mleft\| \mbs{\Phi}_{j, i}(t) \mright\|_2 
    = \sup_{t \in \mathbb{R}_{\geq 0}} \sigma_{j,i}(t) = \mleft\| \sigma_{j,i} \mright\|_{\infty} < \infty,
\end{align} 
for all \(j \in \{u, y\}\) and \( i \in \mathcal{N}\), where \(\sigma_{j,i}(t)\) is the largest singular value of \(\mbs{\Phi}_{j, i}(t)\). This can be thought of as an extension of the \(\mathcal{L}_{\infty}\) space in \Cref{def:Lp_spaces} for piecewise continuous matrix functions \(\mbs{\Phi}_{j, i} : \mathbb{R}_{\geq 0} \to \mathbb{R}^{n_j \times n_j}\), for \(j \in \{u, y\}\) and \( i \in \mathcal{N}\). Finally, for simplicity, the scheduling matrices are assumed to be explicitly time-dependent, though a similar analysis can be used to account for dependence on other signals or parameters.

\subsection{Scheduling Matrix Construction} \label{subsec:scheduling_matrix_construction}
This subsection presents the construction of the scheduling matrices such that they satisfy the following pseudo-commutativity condition.
\begin{definition}[Pseudo-Commuting Scheduling Matrices]\label{def:pseudo_commutativity}
    Consider a gain-scheduled system of type shown in \Cref{fig:GS_G} with scheduling matrices $\mbs{\Phi}_{u, i}(t)\in \mathbb{R}^{n_u \times n_u}$ and $\mbs{\Phi}_{y, i}(t)\in \mathbb{R}^{n_y \times n_y}$ for $i \in \mathcal{N}$. Given a series of matrices \(\mbf{S}_i \in \mathbb{R}^{n_y \times n_u}\) for $i \in \mathcal{N}$, the scheduling matrices are said to \emph{pseudo-commute with $\mbf{S}_i$}, if they satisfy
    \begin{equation} \label{eqn:commutativity}
        \mbs{\Phi}_{y, i}^{\trans}(t)  \mbf{S}_i = \mbf{S}_i \mbs{\Phi}_{u, i}(t), \quad \forall t \in \mathbb{R}_{\geq0}, \forall i \in \mathcal{N}.
    \end{equation}
\end{definition}
\vspace{5pt}
In \Cref{sec:main_results}, for each QSR-dissipative subsystem $\bm{\mathcal{G}}_{i}$ satisfying \cref{eq:QSR-Dissipativity_i} with $\mbf{Q}_i$, $\mbf{S}_i$, and $\mbf{R}_i$, this pseudo-commutativity condition is invoked by assuming the corresponding scheduling matrices of $\bm{\mathcal{G}}_{i}$ pseudo-commute with $\mbf{S}_i$. It will now be shown that the scheduling matrices can be constructed to satisfy the pseudo-commutativity condition in \cref{eqn:commutativity} for any matrix $\mbf{S}_i \in \mathbb{R}^{n_y \times n_u}$.

Consider an arbitrary matrix \(\mbf{S}_i \in \mathbb{R}^{n_y \times n_u}\) and the corresponding input scheduling matrix $\mbs{\Phi}_{u, i}(t)\in \mathbb{R}^{n_u \times n_u}$ and output scheduling matrix $\mbs{\Phi}_{y, i}(t)\in \mathbb{R}^{n_y \times n_y}$.
For the trivial case of \(\mbf{S}_i = \mbf{0}\), the $n_u^2 + n_y^2$ entries of the input and output scheduling matrices can be independently designed to pseudo-commute with \(\mbf{S}_i = \mbf{0}\). Now, consider an arbitrary nonzero matrix \(\mbf{S}_i \in \mathbb{R}^{n_y \times n_u}\) and its singular value decomposition~(SVD) given by
\begin{equation} \label{eqn:SVD_of_S_i}
    \mbf{S}_i = \mbf{U}_i 
    \begin{bmatrix}
        \mbs{\Sigma}_{1, i} & \mbf{0} \\%
        \mbf{0}             & \mbf{0}
    \end{bmatrix}
    \mbf{V}_i^{\trans},
\end{equation}
where \(\mbf{U}_i \in \mathbb{R}^{n_y \times n_y}\), \(\mbs{\Sigma}_{1, i} \in \mathbb{S}^{\varrho_{i}}\), \(\mbf{V}_i \in \mathbb{R}^{n_u \times n_u}\), and \mbox{\(\varrho_{i} = \rank\mleft( \mbf{S}_i \mright) \geq 1\)}. As shown in \Cref{lemma:AS_eq_SB} found in the Appendix, the scheduling matrices $\mbs{\Phi}_{u, i}(t)\in \mathbb{R}^{n_u \times n_u}$ and $\mbs{\Phi}_{y, i}(t)\in \mathbb{R}^{n_y \times n_y}$ satisfy the pseudo-commutativity condition in \cref{eqn:commutativity} if and only if there exists \(\mbf{Z}_{11, i}(t)\), \(\mbf{Z}_{21, i}(t)\), \(\mbf{Z}_{22, i}(t)\), \(\mbf{W}_{21, i}(t)\), and \(\mbf{W}_{22, i}(t)\), with appropriate dimensions, such that
\begin{subequations}\label{eqn:scheduling_matrices_for_general_case}
    \begin{align}
        \mbs{\Phi}_{u, i}(t) &= 
        \mbf{V}_{i} 
        \begin{bmatrix}
            \mbf{Z}_{11, i}(t) & \mbf{0} \\%
            \mbf{Z}_{21, i}(t) & \mbf{Z}_{22, i}(t)
        \end{bmatrix}
        \mbf{V}_{i}^{\trans},
        \\%
        \mbs{\Phi}_{y, i}(t) &= 
        \mbf{U}_{i} 
        \begin{bmatrix}
            \mbs{\Sigma}_{1, i}^{-1} \mbf{Z}_{11, i}^{\trans}(t) \mbs{\Sigma}_{1, i} & \mbf{0}  \\%
            \mbf{W}_{21, i}(t)                                                & \mbf{W}_{22, i}(t)
        \end{bmatrix}
        \mbf{U}_{i}^{\trans}.
    \end{align}
\end{subequations}
Therefore, for any arbitrary matrix \(\mbf{S}_i \in \mathbb{R}^{n_y \times n_u}\) with \mbox{\(\varrho_{i} = \rank\mleft( \mbf{S}_i \mright) \geq 1\)}, the scheduling matrices can be designed using the \(n_{u}^2 + n_{y}^2 + \varrho_{i}^2 - \varrho_{i}(n_u + n_y)\) entries of the independent design variables \(\mbf{Z}_{11, i}(t)\), \(\mbf{Z}_{21, i}(t)\), \(\mbf{Z}_{22, i}(t)\), \(\mbf{W}_{21, i}(t)\), and \(\mbf{W}_{22, i}(t)\) to satisfy the pseudo-commutativity condition in \cref{eqn:commutativity}.

It is also useful to consider the implications of the pseudo-commutativity condition in \cref{eqn:commutativity} for the special case of a square matrix \(\mbf{S}_i \in \mathbb{R}^{n \times n}\). If \(\mbf{S}_i\) is square and full rank, the scheduling matrices can be designed using the similarity transformation
\begin{equation}
    \mbs{\Phi}_{u, i} (t) = \mbf{S}_i^{-1} \mbs{\Phi}_{y, i}^{\trans} (t)\mbf{S}_i, \quad \forall t \in \mathbb{R}_{\geq0}, \forall i \in \mathcal{N},
\end{equation}
where \(\mbs{\Phi}_{y, i}(t) \in \mathbb{R}^{n \times n}\) are free design variables. Additionally, since similarity transformations preserve eigenvalues, the input scheduling matrix can be made full rank by requiring the output scheduling matrix to be full rank. For the case of \(\mbf{S}_i = c_i \eye\) with \(c_i \in \mathbb{R} \setminus \mleft\{ 0 \mright\}\), the pseudo-commutativity in~\cref{eqn:commutativity} simplifies to \(\mbs{\Phi}_{u, i} (t) = \mbs{\Phi}_{y, i}^{\trans} (t)\), which is exactly the scheduling-matrix architecture presented in~\cite{moalemi_forbes_ccta}.
    \section{Main Contribution} \label{sec:main_results}
This section presents the QSR-dissipativity of the matrix-gain-scheduled system in \Cref{fig:GS_G} for two cases.
\begin{itemize}
    \item Case 1: All \(N\) subsystems are QSR-dissipative with \mbox{\(\mbf{Q}_i \in \mathbb{S}_{--}^{n_y}\)}.
    \item Case 2: All \(N\) subsystems are QSR-dissipative with \mbox{\(\mbf{Q}_i \in \mathbb{S}_{-}^{n_y}\)} and share a common \(\mbf{S}_i = \mbf{S} \in \mathbb{R}^{n_y \times n_u}\).
\end{itemize}
Doing so extends the cases discussed in~\cite{QSR} while generalizing the scheduling approach by allowing the scheduling signals to be scheduling matrices. 

    \subsection{Case 1: All \texorpdfstring{\(\mathit{N}\)}{N} Subsystems are QSR-Dissipative with \texorpdfstring{\(\mbf{Q}_\mathit{i} \in \mathbb{S}_{--}^{\mathit{n_y}}\)}{Q < 0}}
    \begin{figure}[t]
        \centering
        \vspace{1pt}
        \includegraphics{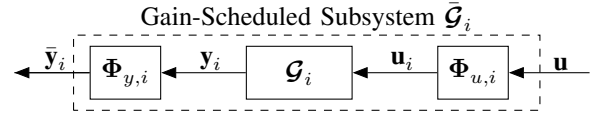}
        \caption{Gain-scheduling of the \(i^{th}\) subsystem, \(\bm{\mathcal{G}}_i\). The input and output of \(\bm{\mathcal{G}}_i\) are scheduled as per \cref{eq:GS_QSR_io} via the matrix multiplication between the scheduling matrices \(\mbs{\Phi}_{u, {\it i}}(t)\) and \(\mbs{\Phi}_{y, {\it i}}(t)\), and their corresponding signals \(\mbf{u}(t)\) and \(\mbf{y}_{{\it i}}(t)\), respectively.}
        \label{fig:GS_Gi}
        \vspace{-5pt}
    \end{figure}
    \begin{lemma} \label{lemma:QSR-Dissipativity_of_subsystem_i}
        Consider the subsystem \(\bm{\mathcal{G}}_i\) in \Cref{fig:GS_Gi} being QSR-dissipative with \(\mbf{Q}_i \in \mathbb{S}_{--}^{n_y}\), \(\mbf{S}_i \in \mathbb{R}^{n_y \times n_u}\), and \(\mbf{R}_i \in \mathbb{S}^{n_u} \). The gain-scheduled subsystem \(\bar{\bm{\mathcal{G}}}_i\) is QSR-dissipative, provided its corresponding scheduling matrices pseudo-commute \mbox{with \(\mbf{S}_i\)}.
    \end{lemma}
    \begin{proof}
        Applying the Rayleigh inequality to the QSR-dissipativity property in \cref{eq:QSR-Dissipativity_i} yields
        \begin{equation} \label{eqn:QSR_per_subsystem_expanded}
            \hspace{-4pt}
            \tilde{V}_i
            \leq
            \lambda_{\max}\mleft( \mbf{Q}_i \mright) \mleft\| \mbf{y}_{i}\mright\|_{2T}^{2} 
            + 2 \left\langle \mbf{y}_i, \mbf{S}_i \mbf{u}_i \right\rangle_T
            + \lambda_{\max}\mleft( \mbf{R}_i \mright) \mleft\| \mbf{u}_{i} \mright\|_{2T}^{2},
        \end{equation}
        where \(\lambda_{\max}\mleft( \mbf{Q}_i \mright) \in \mathbb{R}_{<0}\) and \(\lambda_{\max}\mleft( \mbf{R}_i \mright) \in \mathbb{R}\) since \mbox{\(\mbf{Q}_i \in \mathbb{S}_{--}^{n_y}\)} and \(\mbf{R}_{i} \in \mathbb{S}^{n_y}\), respectively.
        Moreover, based on the pseudo-commutativity condition in \Cref{def:pseudo_commutativity}, it follows that \(\mbs{\Phi}_{y, i}^{\trans}(t) \mbf{S}_i = \mbf{S}_i \mbs{\Phi}_{u, i}(t)\) for all \(t \in \mathbb{R}_{\geq0}\). Consider the term \(\left\langle \mbf{y}_i, \mbf{S}_i \mbf{u}_i \right\rangle_T\) in \cref{eqn:QSR_per_subsystem_expanded}. Using the relation in \cref{eq:GS_QSR_io}, it follows that
        \begin{align}
            \left\langle \mbf{y}_i, \mbf{S}_i \mbf{u}_i \right\rangle_T
            &=
            \left\langle \mbf{y}_i, \mbf{S}_i \mbs{\Phi}_{u, i}\mbf{u} \right\rangle_T
            =
            \left\langle \mbf{y}_i, \mbs{\Phi}_{y, i}^{\trans} \mbf{S}_i \mbf{u} \right\rangle_T \nonumber
            \\%
            &=
            \left\langle \mbs{\Phi}_{y, i} \mbf{y}_i, \mbf{S}_i \mbf{u} \right\rangle_T
            =
            \left\langle \bar{\mbf{y}}_i, \mbf{S}_i \mbf{u} \right\rangle_T. \label{eqn:QSR_yi_to_ybar_for_Si}
        \end{align}
        Additionally, define \(\sigma_{y,i}(t)\) as the largest singular value of \(\mbs{\Phi}_{y, i}(t)\). From \cref{eq:GS_io_y_i}, it follows that
        \begin{align}
            \mleft\| \bar{\mbf{y}}_{i} \mright\|_{2T}^{2}
            =
            \mleft\| \mbs{\Phi}_{y, i} \mbf{y}_{i} \mright\|_{2T}^{2}
            &\leq \int_{0}^{T}\sigma_{y,i}^2(t) \mleft\| \mbf{y}_{i}(t) \mright\|_{2}^{2}\,\dt\nonumber
            \\%
            &\leq
            \bar{\sigma}_{y,i}^2\mleft\| \mbf{y}_{i} \mright\|_{2T}^{2}, \label{eqn:QSR_y_to_ybar_before_rearranging}
        \end{align}
        where 
        \begin{equation} \label{eqn:QSR_bar_sigma_y}
            \bar{\sigma}_{y,i} = \sup_{t \in \mathbb{R}_{\geq0}} \sigma_{y,i}(t).
        \end{equation}
        Furthermore, \(\bar{\sigma}_{y,i} \in \mathbb{R}_{>0} \), provided there exists a \(t \in \mathbb{R}_{\geq0}\) such that \(\mbs{\Phi}_{y, i}(t) \neq \mbf{0} \). This condition is satisfied by design, since if this subsystem is being included in the parallel interconnection, it means at some point in time, its output must be nonzero, and therefore, its output scheduling matrix cannot be zero for all time. Additionally, \(\bar{\sigma}_{y,i} < \infty\), since as per \cref{eqn:scheduling_matrix_boundedness}, the scheduling matrices are assumed to be bounded. Consequently, rearranging \cref{eqn:QSR_y_to_ybar_before_rearranging} yields
        \begin{equation}
            \frac{1}{\bar{\sigma}_{y,i}^2}\mleft\| \bar{\mbf{y}}_{i} \mright\|_{2T}^{2} \leq \mleft\| \mbf{y}_{i} \mright\|_{2T}^{2}. \label{eqn:QSR_y_to_ybar}
        \end{equation}
        Since \(\lambda_{\max}\mleft( \mbf{Q}_i \mright) \in \mathbb{R}_{<0}\), substituting \cref{eqn:QSR_yi_to_ybar_for_Si} and \cref{eqn:QSR_y_to_ybar} into \cref{eqn:QSR_per_subsystem_expanded} leads to
        \begin{equation} \label{eqn:QSR_per_subsystem_expanded_2}
                \tilde{V}_i
                \leq
                \frac{\lambda_{\max}\mleft( \mbf{Q}_i \mright)}{\bar{\sigma}_{y,i}^2} \mleft\| \bar{\mbf{y}}_{i} \mright\|_{2T}^{2}
                + 2\left\langle \bar{\mbf{y}}_{i}, \mbf{S}_i \mbf{u}\right\rangle_T
                + \lambda_{\max}\mleft( \mbf{R}_i \mright) \mleft\| \mbf{u}_{i} \mright\|_{2T}^{2}.
        \end{equation}
        Given that \(\bar{\sigma}_{y,i}^2 \in \mathbb{R}_{>0}\), \cref{eqn:QSR_per_subsystem_expanded_2} can be written as
        \begin{equation} \label{eqn:QSR_per_subsystem_expanded_3}
            \begin{split}
                \bar{\sigma}_{y,i}^2\tilde{V}_i
                &\leq
                -\varepsilon_i \mleft\| \bar{\mbf{y}}_{i} \mright\|_{2T}^{2}
                + 2\bar{\sigma}_{y,i}^2\left\langle \bar{\mbf{y}}_{i}, \mbf{S}_i \mbf{u}\right\rangle_T
                \\%
                &\quad\quad 
                + \lambda_{\max}\mleft( \mbf{R}_i \mright)\bar{\sigma}_{y,i}^2 \mleft\| \mbf{u}_{i} \mright\|_{2T}^{2},
            \end{split}
        \end{equation}
        where 
        \begin{equation} \label{eqn:QSR_per_subsystem_epsilon}
            \varepsilon_i = \lvert \lambda_{\max}\mleft( \mbf{Q}_i \mright) \rvert = - \lambda_{\max}\mleft( \mbf{Q}_i \mright)>0.
        \end{equation}
        Since other than symmetry, there are no restrictions on \(\mbf{R}_i\), the sign of its maximum eigenvalue is unknown. Denoting \(\sigma_{u,i}(t)\) and \(\nu_{u,i}(t)\) as the largest and smallest singular values of \(\mbs{\Phi}_{u, i}(t)\), respectively, define
        \begin{equation} \label{eqn:QSR_delta_i}
            \delta_i = 
            \begin{cases}
                \lambda_{\max}\mleft( \mbf{R}_i \mright) \bar{\sigma}_{y,i}^2\bar{\sigma}_{u,i}^2, 
                & \text{if } \lambda_{\max}\mleft( \mbf{R}_i \mright) > 0, 
                \\
                \lambda_{\max}\mleft( \mbf{R}_i \mright) \bar{\sigma}_{y,i}^2\bar{\nu}_{u,i}^2,    & \text{if } \lambda_{\max}\mleft( \mbf{R}_i \mright) \leq 0,
            \end{cases}
        \end{equation}
        where
        \begin{align}
            \bar{\sigma}_{u,i} &= \sup_{t \in \mathbb{R}_{\geq0}} \sigma_{u,i}(t),
            &
            \bar{\nu}_{u,i} &= \inf_{t \in \mathbb{R}_{\geq0}} \nu_{u,i}(t).
        \end{align}
        If \(\mbf{S}_i\) is full rank and \(n_u \leq n_y\), the pseudo-commutativity property suggests \(\bar{\sigma}_{u,i} \in \mathbb{R}_{>0}\), since  \(\bar{\sigma}_{y,i} \in \mathbb{R}_{>0}\). If \(\mbf{S}_i\) is rank deficient or \(n_y < n_u\), it is possible to have \(\bar{\sigma}_{u,i} \in \mathbb{R}_{\geq 0}\), since the output scheduling matrices can be designed with \(\mbs{\Phi}_{u, i}(t) = \mbf{0}\) for all \(t \in \mathbb{R}_{\geq0}\). Additionally, \(\bar{\nu}_{u,i} \in \mathbb{R}_{\geq 0} \), since the scheduling matrix \(\mbs{\Phi}_{u, i}(t)\) may not be full rank at all times. Consider the term \(\lambda_{\max}\mleft( \mbf{R}_i \mright)\bar{\sigma}_{y,i}^2 \mleft\| \mbf{u}_{i} \mright\|_{2T}^{2}\) in \cref{eqn:QSR_per_subsystem_expanded_3}. Using \cref{eq:GS_io_u_i} and the definition of \(\delta_i\) in \cref{eqn:QSR_delta_i}, it follows that
        \begin{align}
            \lambda_{\max}\mleft( \mbf{R}_i \mright) \bar{\sigma}_{y,i}^2\mleft\| \mbf{u}_{i} \mright\|_{2T}^{2}
            &=
            \lambda_{\max}\mleft( \mbf{R}_i \mright) \bar{\sigma}_{y,i}^2\mleft\| \mbs{\Phi}_{u, i} \mbf{u} \mright\|_{2T}^{2} \nonumber
            \\%
            &\leq
            \delta_i \mleft\| \mbf{u} \mright\|_{2T}^{2}. \label{eqn:QSR_ui_to_u_c}
        \end{align}
        Substituting \cref{eqn:QSR_ui_to_u_c} into \cref{eqn:QSR_per_subsystem_expanded_3} yields
        \begin{align} 
           \hat{V}_i
            &\leq
            -\varepsilon_i\mleft\| \bar{\mbf{y}}_{i} \mright\|_{2T}^{2}
            + 2\bar{\sigma}_{y,i}^2 \left\langle \bar{\mbf{y}}_{i}, \mbf{S}_i \mbf{u}\right\rangle_T 
            + \delta_i \mleft\| \mbf{u} \mright\|_{2T}^{2} \label{eqn:QSR_per-subsystem_norm_form}
            \\%
            &=
            \left\langle \bar{\mbf{y}}_{i}, \bar{\mbf{Q}}_i \bar{\mbf{y}}_{i} \right\rangle_T
            + 
            2 \left\langle \bar{\mbf{y}}_{i}, \bar{\mbf{S}}_i \mbf{u} \right\rangle_T
            + 
            \left\langle \mbf{u}, \bar{\mbf{R}}_i \mbf{u} \right\rangle_T,\nonumber
        \end{align}
        where \(\hat{V}_i = \bar{\sigma}_{y,i}^2\tilde{V}_i\), and
        \begin{align} \label{eqn:QSR_per_subsystem_Qi_Si_Ri}
            \bar{\mbf{Q}}_i &= -\varepsilon_i\eye,
            &
            \bar{\mbf{S}}_i &= \bar{\sigma}_{y,i}^2 \mbf{S}_i,
            &
            \bar{\mbf{R}}_i &= \delta_i \eye.
        \end{align}
        Consequently, the gain-scheduled subsystem \(\bar{\bm{\mathcal{G}}}_i\) is QSR-dissipative with \(\bar{\mbf{Q}}_i \in \mathbb{S}_{--}^{n_y}\), \(\bar{\mbf{S}}_i \in \mathbb{R}^{n_y \times n_u}\), and \(\bar{\mbf{R}}_i \in \mathbb{S}^{n_u}\) as defined in \cref{eqn:QSR_per_subsystem_Qi_Si_Ri}.
    \end{proof}
    \begin{theorem} \label{theorem:1}
        Given that each subsystem \(\bm{\mathcal{G}}_i\) for \( i \in \mathcal{N}\) is QSR-dissipative with \(\mbf{Q}_i \in \mathbb{S}_{--}^{n_y}\), \(\mbf{S}_i \in \mathbb{R}^{n_y \times n_u}\), and \(\mbf{R}_i \in \mathbb{S}^{n_u} \), the gain-scheduled system $\bar{\bm{\mathcal{G}}}$ in \Cref{fig:GS_G} is QSR-dissipative, provided the scheduling matrices \(\mbs{\Phi}_{u, i}(t)\) and \(\mbs{\Phi}_{y, i}(t)\) pseudo-commute with \(\mbf{S}_i\).
    \end{theorem}
    \begin{proof}
        Following \Cref{lemma:QSR-Dissipativity_of_subsystem_i}, each scheduled subsystem \(\bar{\bm{\mathcal{G}}}_i\) is QSR-dissipative with \(\bar{\mbf{Q}}_i \in \mathbb{S}_{--}^{n_y}\), \(\bar{\mbf{S}}_i  \in \mathbb{R}^{n_y \times n_u}\), and \(\bar{\mbf{R}}_i \in \mathbb{S}^{n_u}\) defined in \cref{eqn:QSR_per_subsystem_Qi_Si_Ri}. Furthermore, \(\bar{\sigma}_{y,i}\), \(\varepsilon_i\), and \(\delta_i\) are given by \cref{eqn:QSR_bar_sigma_y}, \cref{eqn:QSR_per_subsystem_epsilon}, and \cref{eqn:QSR_delta_i}, respectively. Adding and subtracting \(\bar{\sigma}_{y,i}^4/\varepsilon_i \mleft\|\mbf{S}_i\mbf{u}\mright\|_{2T}^{2}\) from \cref{eqn:QSR_per-subsystem_norm_form} and factoring leads to
        \begin{equation}\label{eqn:QSR_with_different_S_single_i_before_merging}
            \hat{V}_i
                \leq
                - \varepsilon_i\mleft\|\bar{\mbf{y}}_{i} -  \frac{\bar{\sigma}_{y,i}^2}{\varepsilon_i} \mbf{S}_i\mbf{u} \mright\|_{2T}^{2}\hspace{-5pt}
                + \delta_i \mleft\| \mbf{u} \mright\|_{2T}^{2}
                + \frac{\bar{\sigma}_{y,i}^4}{\varepsilon_i} \mleft\|\mbf{S}_i\mbf{u}\mright\|_{2T}^{2}.
        \end{equation}
        To simplify further analysis, let 
        \begin{equation}
            \tilde{\mbf{y}}_i(t) = \bar{\mbf{y}}_{i}(t) -  \frac{\bar{\sigma}_{y,i}^2}{\varepsilon_i} \mbf{S}_i\mbf{u}(t), \label{eqn:QSR_y_tilde}
        \end{equation}
        for \(i \in \mathcal{N}\) and denote \(\sigma_{\mbf{S}_i}\) as the largest singular value of \(\mbf{S}_i\). Substituting \cref{eqn:QSR_y_tilde} into \cref{eqn:QSR_with_different_S_single_i_before_merging} and using the Rayleigh inequality results in
        \begin{equation}
                \hat{V}_i
                \leq
                - \varepsilon_i\mleft\|\tilde{\mbf{y}}_i \mright\|_{2T}^{2}
                    +
                    \mleft( 
                        \delta_i
                        +
                        \frac{\bar{\sigma}_{y,i}^4\sigma_{\mbf{S}_i}^2}{\varepsilon_i}
                    \mright)
                    \mleft\|\mbf{u}\mright\|_{2T}^{2}. \label{eqn:QSR_with_different_S_single_i}
        \end{equation}
        Summing \cref{eqn:QSR_with_different_S_single_i} over all \(i \in \mathcal{N}\), it follows that
        \begin{align}
            \hspace{-5pt}
            \sum_{i \in \mathcal{N}}\hat{V}_i
            &\leq
                - \sum_{i \in \mathcal{N}}\varepsilon_i\mleft\|\tilde{\mbf{y}}_i \mright\|_{2T}^{2}
                + \sum_{i \in \mathcal{N}}
                \mleft( 
                    \delta_i
                    +
                    \frac{\bar{\sigma}_{y,i}^4\sigma_{\mbf{S}_i}^2}{\varepsilon_i}
                \mright)
                \mleft\|\mbf{u}\mright\|_{2T}^{2}\nonumber
            \\%
            &\leq 
                - \varepsilon_{\min}\sum_{i \in \mathcal{N}}\mleft\|\tilde{\mbf{y}}_i\mright\|_{2T}^{2}
                + \bar{\delta}
                \mleft\|\mbf{u}\mright\|_{2T}^{2}, \label{eqn:QSR_with_different_S_before_cauchy}
        \end{align}
        where 
        \begin{align}
            \varepsilon_{\min} &= \min_{i \in \mathcal{N}}\varepsilon_i > 0,
            &
            \bar{\delta} &= \sum_{i \in \mathcal{N}}
            \mleft( 
                \delta_i
                +
                \frac{\bar{\sigma}_{y,i}^4\sigma_{\mbf{S}_i}^2}{\varepsilon_i}
            \mright). \label{eqn:QSR_hat_delta}
        \end{align}
        Rewriting the summation in \cref{eqn:QSR_with_different_S_before_cauchy} as a vector norm yields
        \begin{align} \label{eqn:QSR_with_different_S_before_cauchy_vector_form}
            \sum_{i \in \mathcal{N}}\hat{V}_i
            &\leq
            -\varepsilon_{\min}
            \mleft\|
            \begin{bmatrix}
                \mleft\|\tilde{\mbf{y}}_1\mright\|_{2T} \\%
                \vdots\\%
                \mleft\|\tilde{\mbf{y}}_N \mright\|_{2T}
            \end{bmatrix}
            \mright\|_{2}^{2}
            + \bar{\delta}
            \mleft\|\mbf{u}\mright\|_{2T}^{2}.
        \end{align}
        Applying the inequality in \Cref{lemma:AM-QM}, which is a special case of the Cauchy\textendash{}Schwartz inequality, to \cref{eqn:QSR_with_different_S_before_cauchy_vector_form} results in
        \begin{align} 
            \sum_{i \in \mathcal{N}}\hat{V}_i
            &\leq 
            - \frac{\varepsilon_{\min}}{N}
            \left\langle 
                \begin{bmatrix}
                    1\\%
                    \vdots\\%
                    1
                \end{bmatrix} 
                ,
                \begin{bmatrix}
                    \mleft\|\tilde{\mbf{y}}_1 \mright\|_{2T} \\%
                    \vdots\\%
                    \mleft\|\tilde{\mbf{y}}_N \mright\|_{2T}
                \end{bmatrix} 
            \right\rangle^{2}
            + \bar{\delta} \mleft\|\mbf{u}\mright\|_{2T}^{2} \nonumber
            \\%
            &=
            - \frac{\varepsilon_{\min}}{N} \mleft(\sum_{i \in \mathcal{N}}\mleft\| \tilde{\mbf{y}}_i\mright\|_{2T}\mright)^{2}
            + \bar{\delta} \mleft\|\mbf{u}\mright\|_{2T}^{2}.\label{eqn:_QSR_with_different_S_after_cauchy_before_triangle}
        \end{align}
        Multiplying both sides of \cref{eqn:_QSR_with_different_S_after_cauchy_before_triangle} by \(N\) and applying the triangle inequality yields
        \begin{align}\label{eqn:QSR_with_different_S_after_cauchy_before_expanding}
            N\sum_{i \in \mathcal{N}}\hat{V}_i
            &\leq 
            - \varepsilon_{\min}\mleft\|\sum_{i \in \mathcal{N}}\tilde{\mbf{y}}_i \mright\|_{2T}^{2} 
            + N\bar{\delta} \mleft\|\mbf{u}\mright\|_{2T}^{2}.
        \end{align}
        Substituting back \cref{eqn:QSR_y_tilde} into \cref{eqn:QSR_with_different_S_after_cauchy_before_expanding} and defining \(\tilde{V} = N\sum_{i \in \mathcal{N}}\hat{V}_i\) results in
        \begin{align}\label{eqn:QSR_with_different_S_after_cauchy_after_expanding}
            \tilde{V}
            &\leq 
            - \varepsilon_{\min}\mleft\|\sum_{i \in \mathcal{N}}\mleft( \bar{\mbf{y}}_{i} - \frac{\bar{\sigma}_{y,i}^2}{\varepsilon_i} \mbf{S}_i\mbf{u} \mright)  \mright\|_{2T}^{2} 
            + N\bar{\delta} \mleft\|\mbf{u}\mright\|_{2T}^{2}.
        \end{align}
        Defining \(\bar{\mbf{S}} = \sum_{i \in \mathcal{N}}\mleft( \bar{\sigma}_{y,i}^2/\varepsilon_i \mright) \mbf{S}_i\) and using \cref{eq:GS_io_y_c}, it follows from~\cref{eqn:QSR_with_different_S_after_cauchy_after_expanding} that
        \begin{align}
            \tilde{V}
            &\leq 
            - \varepsilon_{\min}\mleft\|\mbf{y} -  \bar{\mbf{S}}\mbf{u} \mright\|_{2T}^{2} 
            + N\bar{\delta} \mleft\|\mbf{u}\mright\|_{2T}^{2} \nonumber
            \\%
            &=
            - \varepsilon_{\min}
            \mleft( 
                \mleft\|\mbf{y}\mright\|_{2T}^{2} 
                -2\left\langle \mbf{y}, \bar{\mbf{S}} \mbf{u} \right\rangle_T
                +
                \mleft\|\bar{\mbf{S}}\mbf{u} \mright\|_{2T}^{2}
            \mright)
            + N\bar{\delta} \mleft\|\mbf{u}\mright\|_{2T}^{2}.\label{eqn:QSR_with_different_S_after_cauchy}
        \end{align}
        Applying the Rayleigh inequality to \cref{eqn:QSR_with_different_S_after_cauchy} and defining \(\nu_{\bar{\mbf{S}}}\) as the smallest singular value of \(\bar{\mbf{S}}\), it follows that
        \begin{align}
            \tilde{V}
            &\leq
            - \varepsilon_{\min}\mleft\|\mbf{y}\mright\|_{2T}^{2} 
            + 2\varepsilon_{\min}\left\langle \mbf{y}, \bar{\mbf{S}} \mbf{u} \right\rangle_T
            + \mleft(N\bar{\delta} - \varepsilon_{\min}\nu_{\bar{\mbf{S}}}^{2} \mright) \mleft\|\mbf{u}\mright\|_{2T}^{2} \nonumber
            \\%
            &=
            - \varepsilon_{\min}\mleft\|\mbf{y}\mright\|_{2T}^{2} 
            + 2\varepsilon_{\min}\left\langle \mbf{y}, \bar{\mbf{S}} \mbf{u} \right\rangle_T
            + \hat{\delta}\mleft\|\mbf{u}\mright\|_{2T}^{2}, \label{eqn:QSR_with_different_S_after_cauchy_with_delta_hat}
        \end{align}
        where \(\hat{\delta} = N\bar{\delta} - \varepsilon_{\min}\nu_{\bar{\mbf{S}}}^{2}\).
        Alternatively, \cref{eqn:QSR_with_different_S_after_cauchy_with_delta_hat} can be written as
        \begin{equation*}
            \tilde{V}
            \leq 
            \left\langle \mbf{y}, \mbf{Q} \mbf{y} \right\rangle_T
            + 
            2 \left\langle \mbf{y}, \mbf{S} \mbf{u} \right\rangle_T 
            +
            \left\langle \mbf{u}, \mbf{R} \mbf{u} \right\rangle_T,
        \end{equation*}
        where
        \begin{align} \label{eqn:QSR_Q_S_R_for_case_1}
            \mbf{Q} &= -\varepsilon_{\min} \eye,
            &
            \mbf{S} &= \varepsilon_{\min} \bar{\mbf{S}},
            &
            \mbf{R} &= \hat{\delta} \eye.%
        \end{align}
        Additionally, \cref{eqn:QSR_Q_S_R_for_case_1} can be written in terms of \(\bar{\sigma}_{y,i}\), \(\varepsilon_i\), and \(\delta_i\), given by \cref{eqn:QSR_bar_sigma_y}, \cref{eqn:QSR_per_subsystem_epsilon}, and \cref{eqn:QSR_delta_i}, respectively, as follows:
        \begin{subequations} \label{eqn:QSR_Q_S_R_for_case_1_expanded}
            \begin{align}
                \mbf{Q} &= - \min_{i \in \mathcal{N}}\mleft( \varepsilon_i \mright) \eye,
                \\%
                \mbf{S} &= \min_{i \in \mathcal{N}}\mleft( \varepsilon_i \mright)  \sum_{i \in \mathcal{N}}\frac{\bar{\sigma}_{y,i}^2}{\varepsilon_i} \mbf{S}_i ,
                \\%
                \mbf{R} &= 
                \mleft( 
                    N
                    \sum_{i \in \mathcal{N}}
                    \mleft( 
                        \delta_i
                        +
                        \frac{\bar{\sigma}_{y,i}^4\sigma_{\mbf{S}_i}^2}{\varepsilon_i}
                    \mright)
                    -
                    \min_{i \in \mathcal{N}}\mleft( \varepsilon_i \mright)
                    \nu_{\bar{\mbf{S}}}^{2}
                \mright) \eye,%
            \end{align}
        \end{subequations}
        where \(\sigma_{\mbf{S}_i}\) is the largest singular value of \(\mbf{S}_i\) and \(\nu_{\bar{\mbf{S}}}\) is the smallest singular value of \(\bar{\mbf{S}} = \sum_{i \in \mathcal{N}}\mleft(\bar{\sigma}_{y,i}^2/\varepsilon_i \mright) \mbf{S}_i\). Therefore, the gain-scheduled system \(\bar{\bm{\mathcal{G}}}\) is QSR-dissipative with \mbox{\(\mbf{Q} \in \mathbb{S}_{--}^{n_y}\)}, \(\mbf{S} \in \mathbb{R}^{n_y \times n_u}\), and \(\mbf{R} \in \mathbb{S}^{n_u}\) as defined in \cref{eqn:QSR_Q_S_R_for_case_1_expanded}.
    \end{proof}
\subsection{Case 2: All \texorpdfstring{\(\mathit{N}\)}{N} Subsystems are QSR-Dissipative with \texorpdfstring{\(\mbf{Q}_\mathit{i} \in \mathbb{S}_{-}^{\mathit{n_y}}\)}{Q <= 0} and Share a Common \texorpdfstring{\(\mbf{S}_\mathit{i} = \mbf{S} \in \mathbb{R}^{\mathit{n_y} \times \mathit{n_u}}\)}{Si}}
    \begin{lemma} \label{lemma:QSR-Dissipative_of_GS_wrt_Q}
        Given each subsystem \(\bm{\mathcal{G}}_i\) for \( i \in \mathcal{N}\) is QSR-dissipative with \(\mbf{Q}_i \in \mathbb{S}_{-}^{n_y}\), there exists a \(\mbf{Q} = -\varepsilon \eye\) with \(\varepsilon \in \mathbb{R}_{\geq 0} \) such that the gain-scheduled system $\bar{\bm{\mathcal{G}}}$ in \Cref{fig:GS_G} satisfies 
        \(\sum_{i \in \mathcal{N}} \left\langle \mbf{y}_i, \mbf{Q}_i \mbf{y}_i \right\rangle_T \leq \left\langle \mbf{y}, \mbf{Q} \mbf{y} \right\rangle_T\), provided the output scheduling matrices are active.
    \end{lemma}
    \begin{proof}
        Using the Rayleigh inequality, it follows that
        \begin{align} 
            \sum_{i \in \mathcal{N}} \left\langle \mbf{y}_i, \mbf{Q}_i \mbf{y}_i \right\rangle_T
            &\leq
            -\sum_{i \in \mathcal{N}}\varepsilon_i \mleft\| \mbf{y}_{i} \mright\|_{2T}^{2}\nonumber\\%
            &\leq
            -\varepsilon_{\min}\sum_{i \in \mathcal{N}} \mleft\| \mbf{y}_{i} \mright\|_{2T}^{2}, \label{eqn:R_S_0_tookout_Q_after_sum}
        \end{align} 
        where \(\varepsilon_i = - \lambda_{\max}\mleft( \mbf{Q}_i \mright) \geq 0\) and \(\varepsilon_{\min} = \min_{i \in \mathcal{N}} \varepsilon_i \geq 0\).
        Additionally, the summation in \cref{eq:GS_io_y_c} can be rewritten as the matrix-vector multiplication \(\mbf{y}(t) = \bm{\Psi}(t) \mbs{\upsilon}(t)\), where
        \begin{align}\label{eqn:augmented_matrcies_for_Q}
            \mbs{\Psi}(t) &= 
            \begin{bmatrix}
                \mbs{\Phi}_{y, 1}(t) \,\, \cdots \,\, \mbs{\Phi}_{y, N}(t)
            \end{bmatrix} 
            ,&
            \mbs{\upsilon}(t) &= 
            \begin{bmatrix}
                \mbf{y}_1(t) \\%
                \vdots       \\%
                \mbf{y}_N(t) \\%
            \end{bmatrix}.
        \end{align}
        Consequently, by defining \(\sigma_{\mbs{\Psi}}(t)\) as the largest singular value of \(\mbs{\Psi}(t)\) and applying the Rayleigh inequality to \(\mleft\| \bm{\Psi}(t) \mbs{\upsilon}(t)\mright\|_{2}^{2}\), it follows that
        \begin{align}
            \mleft\| \mbf{y}(t)\mright\|_{2}^{2}
            = 
            \mleft\| \bm{\Psi}(t) \mbs{\upsilon}(t)\mright\|_{2}^{2}
            \leq \sigma_{\mbs{\Psi}}^{2}(t)
            \mleft\| \mbs{\upsilon}(t)\mright\|_{2}^{2}. \label{eqn:QSR_yc_to_yi}
        \end{align}
        Moreover, \(\sigma_{\mbs{\Psi}}(t)\in \mathbb{R}_{>0}\), provided the output scheduling matrices are active, that is, \(\forall t \in \mathbb{R}_{\geq0}\), \(\exists i \in \mathcal{N}\) such that \(\mbs{\Phi}_{y, i}(t) \neq \mbf{0} \), and therefore, \(\bm{\Psi}^{\trans}(t) \bm{\Psi}(t) \in \mathbb{S}_{+}^{N \times n_y}\) is nonzero.
        Additionally,
        \begin{align*}
            \sigma_{\mbs{\Psi}}^2(t) 
            \leq \trace{\mleft( \bm{\Psi}^{\trans} (t) \bm{\Psi}(t) \mright) }
            &= 
                \sum_{i \in \mathcal{N}} \trace{\mleft( \mbs{\Phi}_{y, i}^{\trans} (t) \mbs{\Phi}_{y, i}(t)\mright)}
            \\%
            &\leq 
                n_y \sum_{i \in \mathcal{N}} \lambda_{\max}\mleft( \mbs{\Phi}_{y, i}^{\trans} (t) \mbs{\Phi}_{y, i}(t) \mright)\\%
            &= n_y \sum_{i \in \mathcal{N}} \sigma_{y, i}^2(t).
        \end{align*}
        Therefore, \(\sup_{t \in \mathbb{R}_{\geq0}}\sigma_{\mbs{\Psi}}^2(t) < \infty\) since the output scheduling matrices are assumed to be bounded as per \cref{eqn:scheduling_matrix_boundedness}.
        Consequently, rearranging \cref{eqn:QSR_yc_to_yi} yields
        \begin{equation}
                \frac{1}{\sigma_{\mbs{\Psi}}^{2}(t)} \mleft\| \mbf{y}(t)\mright\|_{2}^{2}
                \leq 
                \mleft\|\mbs{\upsilon}(t)\mright\|_{2}^{2} 
                = \sum_{i \in \mathcal{N}} \mleft\| \mbf{y}_{i}(t) \mright\|_{2}^{2}. \label{eqn:QSR_yc_to_yi_2}
        \end{equation}
        Since \(\lambda_{\max}\mleft( \mbf{Q}_i \mright) \in \mathbb{R}_{\leq 0}\), substituting \cref{eqn:QSR_yc_to_yi_2} into \cref{eqn:R_S_0_tookout_Q_after_sum}, it follows that
        \begin{align} 
            \sum_{i \in \mathcal{N}} \left\langle \mbf{y}_i, \mbf{Q}_i \mbf{y}_i \right\rangle_T
            &\leq
            -\varepsilon_{\min}\int_{0}^{T}\frac{1}{\sigma_{\mbs{\Psi}}^{2}(t)}\mleft\| \mbf{y}(t)\mright\|_{2}^{2}\,\dt\nonumber
            \\%
            &\leq
            -\frac{\varepsilon_{\min}}{\bar{\sigma}_{\mbs{\Psi}}^{2}} \int_{0}^{T}\mleft\| \mbf{y}(t)\mright\|_{2}^{2} \,\dt\nonumber
            \\%
            &= - \varepsilon \mleft\| \mbf{y}\mright\|_{2T}^{2}
            = \left\langle \mbf{y}, \mbf{Q} \mbf{y} \right\rangle_T,\nonumber
        \end{align} 
        where
        \begin{align}\label{eqn:Q_epsilon}
        \bar{\sigma}_{\mbs{\Psi}} &= \sup_{t \in \mathbb{R}_{\geq0}} \sigma_{\mbs{\Psi}}(t) > 0, &
        \varepsilon &= \frac{\varepsilon_{\min}}{\bar{\sigma}_{\mbs{\Psi}}^2} \geq 0,
        \end{align}
        and \(\mbf{Q} = -\varepsilon \eye \in \mathbb{S}_{-}^{n_y}\).
    \end{proof}
    \begin{lemma} \label{lemma:QSR-Dissipative_of_GS_wrt_R}
        Given each subsystem \(\bm{\mathcal{G}}_i\) for \( i \in \mathcal{N}\) is QSR-dissipative with \(\mbf{R}_i \in \mathbb{S}^{n_u}\), there exists an \(\mbf{R} = \delta \eye\) with \(\delta \in \mathbb{R} \) such that the gain-scheduled system, $\bar{\bm{\mathcal{G}}}$, in \Cref{fig:GS_G} satisfies 
        \(\sum_{i \in \mathcal{N}} \left\langle \mbf{u}_i, \mbf{R}_i \mbf{u}_i \right\rangle_T \leq \left\langle \mbf{u}, \mbf{R} \mbf{u} \right\rangle_T\).
    \end{lemma}
    \begin{proof}
        Using the Rayleigh inequality, it follows that
        \begin{align} \label{eqn:Q_S_0_tookout_R_sum}
            \sum_{i \in \mathcal{N}} \left\langle \mbf{u}_i, \mbf{R}_i \mbf{u}_i \right\rangle_T
            &\leq
            \sum_{i \in \mathcal{N}} \lambda_{\max}\mleft( \mbf{R}_i \mright) \mleft\| \mbf{u}_{i}\mright\|_{2T}^{2}.
        \end{align}
        Defining the index sets
        \begin{subequations} \label{eqn:R_sets}
            \begin{align}
                \mathcal{R}_{>0} &= \mleft\{\,i \in \mathcal{N} \mid \lambda_{\max}\mleft( \mbf{R}_i \mright) > 0 \,\mright\},
                \\%
                \mathcal{R}_{<0} &= \mleft\{\,i \in \mathcal{N} \mid \lambda_{\max}\mleft( \mbf{R}_i \mright) < 0 \,\mright\},
                \\%
                \mathcal{R}_{0}  &= \mleft\{\,i \in \mathcal{N} \mid \lambda_{\max}\mleft( \mbf{R}_i \mright) = 0 \,\mright\},
            \end{align}
        \end{subequations}
        it follows that
        \begin{equation} \label{eqn:R_sets_sum}
            \sum_{i \in \mathcal{N}} \lambda_{\max}\mleft( \mbf{R}_i \mright)
            =
            \sum_{i \in \mathcal{R}_{>0}} \lambda_{\max}\mleft( \mbf{R}_i \mright)
            +
            \sum_{i \in \mathcal{R}_{<0}} \lambda_{\max}\mleft( \mbf{R}_i \mright).
        \end{equation}
        Expanding \cref{eqn:Q_S_0_tookout_R_sum} using \cref{eqn:R_sets_sum} results in
        \begin{equation} \label{eqn:R_sets_sum_expanded}
            \begin{split}
                \sum_{i \in \mathcal{N}} \left\langle \mbf{u}_i, \mbf{R}_i \mbf{u}_i \right\rangle_T
                \leq
                &\sum_{i \in \mathcal{R}_{>0}} \lambda_{\max}\mleft( \mbf{R}_i \mright) \mleft\| \mbf{u}_i\mright\|_{2T}^{2} 
                \\
                    \quad &+
                    \sum_{i \in \mathcal{R}_{<0}} \lambda_{\max}\mleft( \mbf{R}_i \mright) \mleft\|\mbf{u}_i\mright\|_{2T}^{2}.
            \end{split}
        \end{equation}
        Substituting the relation in \cref{eq:GS_io_u_i} into \cref{eqn:R_sets_sum_expanded}, it follows that
        \begin{equation} \label{eqn:R_sets_sum_expanded_2}
            \begin{split}
            \sum_{i \in \mathcal{N}} \left\langle \mbf{u}_i, \mbf{R}_i \mbf{u}_i \right\rangle_T
            \leq
            {}&\delta_{\max} \sum_{i \in \mathcal{R}_{>0}}\mleft\| \mbs{\Phi}_{u, i} \mbf{u}\mright\|_{2T}^{2} 
            \\
            \quad &-
            \delta_{\min}\sum_{i \in \mathcal{R}_{<0}} \mleft\| \mbs{\Phi}_{u, i} \mbf{u}\mright\|_{2T}^{2},
            \end{split}
        \end{equation}
        \vspace{-3pt}
        where 
            \begin{align*}
                \delta_i &= \lvert \lambda_{\max}\mleft( \mbf{R}_i \mright) \rvert,
                &
                \delta_{\max} &= \max_{i \in \mathcal{R}_{>0}} \delta_i,
                &
                \delta_{\min} &= \min_{i \in \mathcal{R}_{<0}} \delta_i.
            \end{align*}
        Expanding \cref{eqn:R_sets_sum_expanded_2} and using the Rayleigh inequality yields
        \begin{align}
            \sum_{i \in \mathcal{N}}\left\langle \mbf{u}_i, \mbf{R}_i \mbf{u}_i \right\rangle_T
            &\leq
            \delta_{\max} \sum_{i \in \mathcal{R}_{>0}}  \int_{0}^{T} \sigma_{u,i}^2(t) \mleft\| \mbf{u}(t)\mright\|_{2}^{2}\,\dt \nonumber
            \\%
            &\quad
            -\delta_{\min} \sum_{i \in \mathcal{R}_{<0}}  \int_{0}^{T} \nu_{u,i}^2(t) \mleft\| \mbf{u}(t)\mright\|_{2}^{2}\,\dt\nonumber
            \\%
            &\leq
                \delta_{\max}\bar{\sigma}_{u} \int_{0}^{T} \mleft\| \mbf{u}(t)\mright\|_{2}^{2}\,\dt\nonumber
                \\%
                &\quad
                -\delta_{\min} \bar{\nu}_{u}\int_{0}^{T} \mleft\| \mbf{u}(t)\mright\|_{2}^{2}\,\dt\nonumber
            \\%
            &= \delta \mleft\| \mbf{u}\mright\|_{2T}^{2}= \left\langle \mbf{u}, \mbf{R} \mbf{u} \right\rangle_T,\nonumber
        \end{align}
        where
        \begin{subequations} \label{eqn:QSR_R}
            \begin{gather}
                \bar{\sigma}_{u} =\sup_{t \in \mathbb{R}_{\geq0}} \sum_{i \in \mathcal{R}_{>0}} \sigma_{u,i}^2(t) \geq  0, \label{eqn:sigma_bar_u}\\%
                \bar{\nu}_{u} = \inf_{t \in \mathbb{R}_{\geq0}} \sum_{i \in \mathcal{R}_{<0}} \nu_{u,i}^2(t) \geq  0
                ,\\%
                \delta = \delta_{\max}\bar{\sigma}_{u} -\delta_{\min} \bar{\nu}_{u}
                , \label{eq:QSR_delta}
            \end{gather}
        \end{subequations}
        and \(\mbf{R} = \delta \eye\).
    \end{proof}

    Since \Cref{lemma:QSR-Dissipative_of_GS_wrt_R} considers the most general case of subsystems being QSR-dissipative with \(\mbf{R}_i \in \mathbb{S}^{n_u}\) and imposes no further restriction on the eigenvalues of \(\mbf{R}_i\), the sign of \(\delta\) in~\cref{eq:QSR_delta} cannot be determined without additional information.
    \begin{corollary} \label{col:QSR-Dissipative_of_GS_wrt_R}
        Consider the gain-scheduled system, $\bar{\bm{\mathcal{G}}}$, in \Cref{fig:GS_G} where each subsystem \(\bm{\mathcal{G}}_i\) for \( i \in \mathcal{N}\) is QSR-dissipative. Given the sets \(\mathcal{R}_{>0}\), \(\mathcal{R}_{<0}\), and \(\mathcal{R}_{0}\) defined in \cref{eqn:R_sets}, the special cases of \Cref{lemma:QSR-Dissipative_of_GS_wrt_R} are as follows:
        \begin{enumerate}
            \item{%
                \(\delta \in \mathbb{R}_{\geq 0}\) provided \(\mathcal{R}_{>0} \neq \varnothing\) and \(\mathcal{R}_{<0} = \varnothing\).
            }%
            \begin{enumerate}[1.1)]
                \item{%
                    \(\delta \in \mathbb{R}_{>0}\), provided further \(\exists t \in \mathbb{R}_{\geq0} \) and an \mbox{\(i \in \mathcal{R}_{>0}\)} such that \(\mbs{\Phi}_{u, i}(t) \neq \mbf{0} \).
                }%
            \end{enumerate}
            \item{%
                \(\delta \in \mathbb{R}_{\leq 0}\) provided \(\mathcal{R}_{>0} = \varnothing\) and \(\mathcal{R}_{<0} \neq \varnothing\).
            }%
            \begin{enumerate}[2.1)]
                \item{%
                    \(\delta \in \mathbb{R}_{<0}\), provided further \(\mathcal{R}_{<0} \cap \mathcal{F}_u(t) \neq \varnothing\), \mbox{\(\forall t \in \mathbb{R}_{\geq0}\)}, meaning for \(i \in \mathcal{R}_{<0}\), the nonempty subset of input scheduling matrices, \(\mbs{\Phi}_{u, i}(t)\), is strongly active.
                }%
            \end{enumerate}
            \item{%
                \(\delta = 0\) provided \(\mathcal{R}_{>0} = \varnothing\) and \(\mathcal{R}_{<0} = \varnothing\).
            }%
        \end{enumerate}
    \end{corollary}
    \begin{proof}
        The proof for each case follows from the definition of \(\delta\) in \cref{eq:QSR_delta} and the sets \(\mathcal{R}_{>0}\), \(\mathcal{R}_{<0}\), and \(\mathcal{R}_{0}\) in \cref{eqn:R_sets}. The strictly positive case follows from the fact that \( \bar{\sigma}_{u} \in \mathbb{R}_{>0}\), provided there exists a \(t \in \mathbb{R}_{\geq0}\) and an \(i \in \mathcal{R}_{>0}\) such that \(\mbs{\Phi}_{u, i}(t) \neq \mbf{0} \). 
        The strictly negative case follows from
        \begin{equation*}
            \bar{\nu}_{u} = \inf_{t \in \mathbb{R}_{\geq0}} \sum_{i \in \mathcal{R}_{<0}} \nu_{u,i}^{2}(t) = \inf_{t \in \mathbb{R}_{\geq0}} \sum_{i \in \mathcal{R}_{<0} \cap \mathcal{F}_u(t)} \nu_{u,i}^{2}(t) > 0,
        \end{equation*}
        provided that \(\mathcal{R}_{<0} \cap \mathcal{F}_u(t) \neq \varnothing\) for all \(t \in \mathbb{R}_{\geq0}\).
    \end{proof}
    \begin{theorem} \label{theorem:2}%
        Given each subsystem \(\bm{\mathcal{G}}_i\) for \( i \in \mathcal{N}\) is QSR-dissipative with \(\mbf{Q}_i \in \mathbb{S}_{-}^{n_y}\), \(\mbf{S}_i = \mbf{S} \in \mathbb{R}^{n_y \times n_u} \), and \(\mbf{R}_i \in \mathbb{S}^{n_u}\), the gain-scheduled system, $\bar{\bm{\mathcal{G}}}$, in \Cref{fig:GS_G} is QSR-dissipative, 
        provided the output scheduling matrices are active, and the scheduling matrices pseudo-commute with \(\mbf{S}\).
    \end{theorem}
    \begin{proof}%
        Summing \cref{eq:QSR-Dissipativity_i} over all \(N\) subsystems yields
        \begin{equation}
            \begin{split}
                \sum_{i \in \mathcal{N}}\tilde{V}_i
                &\leq 
                \sum_{i \in \mathcal{N}} \left\langle \mbf{y}_i, \mbf{Q}_i \mbf{y}_i \right\rangle_T
                + 2\sum_{i \in \mathcal{N}}\left\langle \mbf{y}_i, \mbf{S} \mbf{u}_i \right\rangle_T
                \\%
                &\quad\quad
                + \sum_{i \in \mathcal{N}} \left\langle \mbf{u}_i, \mbf{R}_i \mbf{u}_i \right\rangle_T. \label{eqn:th2_start}
            \end{split}
        \end{equation}
        Consider the term \(\sum_{i \in \mathcal{N}} \left\langle \mbf{y}_i, \mbf{S} \mbf{u}_i \right\rangle_T\). Given the scheduling matrices pseudo-commute with respect to \(\mbf{S}\) such that \mbox{\(\mbs{\Phi}_{y, i}^{\trans}(t) \mbf{S} = \mbf{S} \mbs{\Phi}_{u, i}(t)\)}, and using \cref{eq:GS_QSR_io} and \cref{eq:GS_io_y_c}, it follows that
        \begin{align}
            \sum_{i \in \mathcal{N}} \left\langle \mbf{y}_i, \mbf{S} \mbf{u}_i \right\rangle_T 
            &= 
            \sum_{i \in \mathcal{N}} \left\langle \mbf{y}_i, \mbf{S}\mbs{\Phi}_{u, i} \mbf{u}\right\rangle_T \nonumber\\%
            &= 
            \sum_{i \in \mathcal{N}}\left\langle \mbf{y}_i, \mbs{\Phi}_{y, i}^{\trans}\mbf{S} \mbf{u}\right\rangle_T \nonumber
            \\%
            &= 
            \sum_{i \in \mathcal{N}}\left\langle \mbs{\Phi}_{y, i}\mbf{y}_i, \mbf{S} \mbf{u}\right\rangle_T
            = \left\langle \mbf{y}, \mbf{S} \mbf{u}\right\rangle_T. \label{eqn:th2_S_y_to_u}
        \end{align}
        Consequently, using \Cref{lemma:QSR-Dissipative_of_GS_wrt_Q,lemma:QSR-Dissipative_of_GS_wrt_R} and substituting \cref{eqn:th2_S_y_to_u} into \cref{eqn:th2_start} yields 
        \begin{equation*}
            \begin{split}
                \tilde{V} =
                \sum_{i \in \mathcal{N}} \tilde{V}_i
                &\leq 
                \left\langle \mbf{y}, \mbf{Q} \mbf{y} \right\rangle_T
                + 
                2 \left\langle \mbf{y}, \mbf{S} \mbf{u} \right\rangle_T 
                +
                \left\langle \mbf{u}, \mbf{R} \mbf{u} \right\rangle_T,
            \end{split}
        \end{equation*}
        where \(\mbf{Q} = -\varepsilon\eye\) and \(\mbf{R} = \delta\eye\), with \(\varepsilon\) and \(\delta\) defined as per \Cref{eqn:Q_epsilon,eqn:QSR_R}, respectively. Therefore, the gain-scheduled system, $\bar{\bm{\mathcal{G}}}$, is QSR-dissipative with \(\mbf{Q} \in \mathbb{S}_{-}^{n_y}\), \(\mbf{S} \in \mathbb{R}^{n_y \times n_u} \), and \mbox{\(\mbf{R} \in \mathbb{S}^{n_u}\)}.
    \end{proof}
\section{Discussion} \label{Discussion}
In this section, the special cases of QSR-dissipativity, as defined in \Cref{def:QSR-Dissipativity}, are discussed for the gain-scheduled system $\bar{\bm{\mathcal{G}}}$ in \Cref{fig:GS_G}. For each case, to compare the scheduling matrix results presented in \Cref{theorem:1,theorem:2} with the existing literature on scalar scheduling signals, the special case of \(\mbs{\Phi}_{u, i}(t) = s_i(t)\eye\) and \(\mbs{\Phi}_{y, i}(t) = s_i(t)\eye\) for all \(i \in \mathcal{N}\) is considered and referred to as the \emph{base case}. 

Moreover, the matrix-scheduling architecture in~\cite{moalemi_forbes_ccta} includes additional constant scaling parameters, \(\alpha_i \in \mathbb{R}_{>0}\), for each subsystem \(\bm{\mathcal{G}}_i\). These scaling parameters are not considered in this work to maintain focus on the proposed generalized matrix-gain-scheduling architecture. However, they can easily be incorporated into the results presented in \Cref{theorem:1,theorem:2}. Therefore, when comparing the results of this work with~\cite{moalemi_forbes_ccta}, the scaling parameters are assumed to be \(\alpha_i = 1\) for all \(i \in \mathcal{N}\).

\subsection{Subsystems are Passive}
    Consider the case where each subsystem \(\bm{\mathcal{G}}_i\) for \( i \in \mathcal{N}\) is passive with \(\mbf{Q}_i = \mbf{0}\), \(\mbf{S}_i = \frac{1}{2} \eye\), and \(\mbf{R}_i = \mbf{0}\). \Cref{theorem:2} suggests $\bar{\bm{\mathcal{G}}}$ is QSR-dissipative with \(\mbf{Q} = \mbf{0}, \mbf{S} = \frac{1}{2} \eye\), and \mbox{\(\mbf{R} = \mbf{0}\)}, and therefore is passive, provided the scheduling matrices pseudo-commute with \(\mbf{S}\). Note, \Cref{theorem:2} further requires the output scheduling matrices to be active. However, this condition can be relaxed for the passive case since \(\mbf{Q}_i = \mbf{0}\), for all \(i \in \mathcal{N}\). As discussed in \Cref{subsec:scheduling_matrix_construction}, for this choice of \(\mbf{S}\), the pseudo-commutativity condition leads to \mbox{\(\mbs{\Phi}_{u, i}(t) = \mbs{\Phi}_{y, i}^{\trans}(t)\)} for all \(i \in \mathcal{N}\) and \(t \in \mathbb{R}_{\geq0}\). Therefore, the gain-scheduling of passive subsystems using scheduling matrices results in an overall passive gain-scheduled system. For the base case, this result is consistent with the results shown in~\cite[Theorem~8.2]{Marquez}, which is again shown in~\mbox{\cite[Theorem~2]{QSR}} and~\cite[Corollary~6.1.2]{Forbes_thesis}.
\subsection{Subsystems are ISP}
    Consider the case where each subsystem \(\bm{\mathcal{G}}_i\) for \( i \in \mathcal{N}\) is ISP with \(\mbf{Q}_i = \mbf{0}\), \(\mbf{S}_i = \frac{1}{2} \eye\), and \(\mbf{R}_i = -\delta_i \eye\), where \(\delta_i \in \mathbb{R}_{>0}\). \Cref{theorem:2} suggests $\bar{\bm{\mathcal{G}}}$ is QSR-dissipative with 
    \begin{align*}
        \mbf{Q} &= \mbf{0},
        &
        \mbf{S} &= \frac{1}{2} \eye,
        &
        \mbf{R} &= -\delta \eye,
    \end{align*}
    where 
    \begin{align}\label{eqn:QSR_delta_ISP_not_strictly_active}
        \delta 
        & = \min_{i \in \mathcal{N}} \mleft( \delta_i \mright) \inf_{t \in \mathbb{R}_{\geq0}}\sum_{i \in \mathcal{N}} \nu_{u,i}^2(t) \geq 0,
    \end{align}
    provided the scheduling matrices pseudo-commute with \mbox{\(\mbf{S} = \frac{1}{2} \eye\)}. Furthermore, if the scheduling matrices are designed to be strongly active, \Cref{col:QSR-Dissipative_of_GS_wrt_R} with \(\mathcal{R}_{<0} = \mathcal{N}\) suggests $\bar{\bm{\mathcal{G}}}$ is ISP with
    \begin{equation} \label{eqn:QSR_delta_ISP_fully_active}
        \delta = \delta_{\min} \inf_{t \in \mathbb{R}_{\geq0}} \sum_{i \in \mathcal{F}_u(t)} \nu_{u,i}^2(t) > 0,
    \end{equation} 
    where \(\delta_{\min} = \min_{i \in \mathcal{N}} \delta_i\). Note, as per the definition of \(\mathcal{F}_u(t)\) in~\cref{eqn:Fj(t)}, \(s_i(t) = 0\) for \(i \in \mathcal{N} \setminus \mathcal{F}_u(t)\). The ISP coefficient in~\cref{eqn:QSR_delta_ISP_fully_active}, exactly recovers the matrix scheduling of ISP subsystems result shown in~\cite[Theorem~1]{moalemi_forbes_ccta}. For the base case, the ISP coefficient \(\delta\) in \cref{eqn:QSR_delta_ISP_fully_active} simplifies to 
    \begin{align} \label{eqn:QSR_delta_ISP_base_case}
        \delta &= \delta_{\min} \inf_{t \in \mathbb{R}_{\geq0}} \sum_{i \in \mathcal{F}_u(t)} \lvert s_i(t) \rvert^2,
    \end{align}
    matching the results in~\cite[Theorem~1]{Damaren_passive_map}. In~\cite[Theorem~2]{QSR} a slightly different bound on \(\delta\) is provided as
    \begin{align} \label{eqn:ryans_delta_ISP}
        \delta &=  \sum_{i \in \mathcal{N}} \delta_i \nu_{i}^{2},
    \end{align}
    where \(\nu_{i} = \inf_{t \in \mathbb{R}_{\geq0}}  \lvert s_i(t) \rvert \geq 0\) for all \(i \in \mathcal{N}\). The authors of~\cite{QSR} argue that \cref{eqn:ryans_delta_ISP} is a less conservative bound on \(\delta\) compared to~\cite[Theorem~1]{Damaren_passive_map}, since it involves the summation of \(\delta_i\) across all subsystems instead of the minimum value of \(\delta_i\). However, \cref{eqn:ryans_delta_ISP} also involves the summation of infima, which is more conservative than the infimum of the summation in \cref{eqn:QSR_delta_ISP_base_case}. Therefore, it is not guaranteed that \cref{eqn:ryans_delta_ISP} is less conservative than \cref{eqn:QSR_delta_ISP_base_case}. Additionally, to maintain the ISP property of the subsystems,~\cite[Theorem~2]{QSR} requires the existence of at least one \(i \in \mathcal{N}\) such that \(\nu_{i} > 0\). This means that there exists at least one scheduled subsystem \(\bar{\bm{\mathcal{G}}}_i\) with \(i \in \mathcal{N}\) that is always active in the sense that \(s_i(t) \neq 0\) for all \(t \in \mathbb{R}_{\geq0}\). Depending on the application, one may need to freely switch between the scheduled subsystems, which may not be possible if at least one subsystem is always active. In practice, it can be argued that by setting the scheduling signal to an arbitrary small value, the scheduled subsystem can be made inactive while still retaining the ISP property of the gain-scheduled system. 
    
    Contrary to~\cite{QSR}, at any time,~\cite[Theorem~1]{Damaren_passive_map} requires at least one scheduling signal to be nonzero. This ensures that for all time \(t \in \mathbb{R}_{\geq0}\), there exists at least one \(i \in \mathcal{N}\) such that \(\mbf{u}_i(t)\) and \(\bar{\mbf{y}}_i(t)\) are nonzero, provided that \(\mbf{u}(t)\) and \(\mbf{y}_{i}(t)\) are nonzero. Here, the scheduling matrices are required to pseudo-commute with \(\mbf{S}\) and be strongly active to ensure for all time \(t \in \mathbb{R}_{\geq0}\), there exists at least one \(i \in \mathcal{N}\) such that \mbox{\(\mbs{\Phi}_{u, i}(t) = \mbs{\Phi}_{y, i}^{\trans}(t)\)} is full rank. The existence of a full rank scheduling matrix at each time ensures that there exists at least one \(i \in \mathcal{N}\) such that \(\mbf{u}_i(t)\) and \(\bar{\mbf{y}}_i(t)\) are nonzero, provided that \(\mbf{u}(t)\) and \(\mbf{y}_{i}(t)\) are nonzero, which can be thought of as direct extension of the aforementioned condition in~\cite[Theorem~1]{Damaren_passive_map}. 
\subsection{Subsystems are OSP}
    Consider the case where each subsystem \(\bm{\mathcal{G}}_i\) for \( i \in \mathcal{N}\) is OSP with \(\mbf{Q}_i = - \varepsilon_i\eye\), \(\mbf{S}_i = \frac{1}{2} \eye\), and \(\mbf{R}_i = \mbf{0}\), where \(\varepsilon_i \in \mathbb{R}_{>0}\). \Cref{theorem:1} suggests $\bar{\bm{\mathcal{G}}}$ is QSR-dissipative with
    \begin{gather*}
        \mbf{Q} = -\varepsilon_{\min} \eye, \quad\quad
        \mbf{S} = \varepsilon_{\min} \sum_{i \in \mathcal{N}}\frac{\bar{\sigma}_{y,i}^2}{2\varepsilon_i} \eye,
        \\%
        \mbf{R} = 
        \mleft( 
            N
            \sum_{i \in \mathcal{N}}
            \mleft( 
                \frac{\bar{\sigma}_{y,i}^4}{4\varepsilon_i}
            \mright)
            -
            \varepsilon_{\min}
            \mleft( \sum_{i \in \mathcal{N}}\frac{\bar{\sigma}_{y,i}^2}{2\varepsilon_i} \mright) ^{2}
        \mright) \eye,
    \end{gather*}
    provided the scheduling matrices pseudo-commute with \mbox{\(\mbf{S} = \frac{1}{2} \eye\)}. Notice \(\mbf{R} \in \mathbb{S}_{+}^{n_u}\) as a consequence of \Cref{lemma:AM-QM}, since
    \begin{align*}
        N
        \sum_{i \in \mathcal{N}}
            &\frac{\bar{\sigma}_{y,i}^4}{4\varepsilon_i}
            =
            N
            \sum_{i \in \mathcal{N}}
                \varepsilon_{i}
                \frac{\bar{\sigma}_{y,i}^4}{4 \varepsilon_{i}^{2}}
            \geq
            \mleft( 
                \sum_{i \in \mathcal{N}}
                    \sqrt{\varepsilon_{i}}
                    \frac{\bar{\sigma}_{y,i}^2}{2 \varepsilon_{i}}
            \mright)^2\nonumber
            \\%
            &\geq
            \mleft(
                \min_{i \in \mathcal{N}}\mleft( \sqrt{\varepsilon_{i}}\mright) 
                \sum_{i \in \mathcal{N}}
                    \frac{\bar{\sigma}_{y,i}^2}{2 \varepsilon_{i}}
            \mright)^2 \hspace{-4pt}=
            \varepsilon_{\min}
            \mleft(
                \sum_{i \in \mathcal{N}}
                    \frac{\bar{\sigma}_{y,i}^2}{2 \varepsilon_{i}}
            \mright)^2,
    \end{align*}
    with \(\mbf{R} = \mbf{0}\) if and only if \(\bar{\sigma}_{y,i}^2 /\sqrt{\varepsilon_{i}} = k\), for some positive constant \(k \in \mathbb{R}_{>0}\) and for all \(i \in \mathcal{N}\). Therefore, \Cref{theorem:1} does not provide a tight bound on \(\mbf{R}\) for the OSP case. Similarly, for the base case,~\cite[Theorem~1]{QSR} also leads to \(\mbf{R} \in \mathbb{S}_{+}^{n_u}\), with \(\mbf{R} = \mbf{0}\) if and only if \(\sup_{t \in \mathbb{R}_{\geq0}} \lvert s_i(t) \rvert^{2} /\varepsilon_i = k\), for some positive constant \(k \in \mathbb{R}_{>0}\), for all \(i \in \mathcal{N}\). However, \Cref{theorem:2} suggests $\bar{\bm{\mathcal{G}}}$ is QSR-dissipative with 
        \begin{align*}
            \mbf{Q} &= - \varepsilon \eye,&
            \mbf{S} &= \frac{1}{2} \eye, &
            \mbf{R} &= \mbf{0},
        \end{align*}
    and therefore is OSP with 
    \begin{equation} \label{eqn:QSR_epsilon_OSP}
        \varepsilon = \frac{\varepsilon_{\min}}{\bar{\sigma}_{\mbs{\Psi}}^2} > 0,
    \end{equation}
    provided the output scheduling matrices are active and the scheduling matrices pseudo-commute with \(\mbf{S} = \frac{1}{2} \eye\). Additionally, for the base case, the augmented matrix, \(\mbs{\Psi}(t)\), defined in \cref{eqn:augmented_matrcies_for_Q}, can be written as
    \begin{align*}
        \mbs{\Psi}(t) &= 
        \begin{bmatrix}
            s_1(t) & \cdots & s_N(t)
        \end{bmatrix}
        \bm{\otimes} \eye 
        =
        \mbf{s}^{\trans}(t) \bm{\otimes} \eye,
    \end{align*}
    where \(\bm{\otimes}\) denotes the Kronecker product. Since the output scheduling matrices are assumed to be active, then \(\mbf{s}(t) \neq \mbf{0}\) for all \(t \in \mathbb{R}_{\geq0}\). Therefore, the nonzero singular values of \(\mbs{\Psi}(t)\) are the positive numbers \(\mleft\{\,\sigma(\mbf{s}(t))\sigma_i(\eye)
    \mid i \in \mathcal{N}_y \,\mright\}\), where \mbox{\(\mathcal{N}_y = \mleft\{ 1, 2, \ldots, n_y \mright\}\)}~\cite[Theorem~4.2.15]{Horn_Johnson_1991}. Consequently,
    \begin{equation*}
        \sigma_{\mbs{\Psi}}(t) =  \max_{i \in \mathcal{N}_y}  \sigma\mleft( \mbf{s}(t) \mright)\sigma_i(\eye)
        = \sqrt{\sum_{i \in \mathcal{N}} \lvert s_i(t) \rvert ^2} > 0.
    \end{equation*}
    From \cref{eqn:Q_epsilon}, it follows that
    \begin{equation} \label{eqn:QSR_sigma_bar_Psi_for_scalar_si}
        \bar{\sigma}_{\mbs{\Psi}}^2 = \sup_{t \in \mathbb{R}_{\geq0}} \sum_{i \in \mathcal{N}} \lvert s_i(t) \rvert ^2 > 0.
    \end{equation}
    Similar to the ISP case, the OSP coefficient, \(\varepsilon \), in \cref{eqn:QSR_epsilon_OSP} exactly matches matrix scheduling of OSP subsystems shown in~\mbox{\cite[Theorem~2]{moalemi_forbes_ccta}}. For the base case, \cref{eqn:QSR_sigma_bar_Psi_for_scalar_si} further expands on the relationship between the scheduling signals and the OSP coefficient
    \begin{equation} \label{eqn:QSR_epsilon_OSP_base_case}
        \varepsilon = \frac{\varepsilon_{\min}}{\sup_{t \in \mathbb{R}_{\geq0}} \sum_{i \in \mathcal{N}} \lvert s_i(t) \rvert ^2} > 0.
    \end{equation}
    Using the Cauchy\textendash{}Schwartz inequality, it can be shown that if a system is OSP with \(\mbf{Q} = - \varepsilon\eye\), it also possesses finite \(\mathcal{L}_2\) gain such that \(\varepsilon = 1/\gamma\). Here, by defining \(\varepsilon_i = 1/\gamma_i\) with \(\gamma_i \in \mathbb{R}_{>0}\), \cref{eqn:QSR_epsilon_OSP_base_case} can be rewritten as \(\varepsilon = 1/\gamma\), where
    \begin{equation} \label{eqn:QSR_finite_L2_scalar_scheduling_gamma}
        \gamma = \max_{i \in \mathcal{N}} \mleft( \gamma_i \mright)\sup_{t \in \mathbb{R}_{\geq0}} \sum_{i \in \mathcal{N}} \lvert s_i(t) \rvert ^2 > 0.
    \end{equation}
    Therefore, the gain-scheduled system $\bar{\bm{\mathcal{G}}}$ possess finite \(\mathcal{L}_2\) gain with \(\gamma\) defined as per \cref{eqn:QSR_finite_L2_scalar_scheduling_gamma}.
\subsection{Subsystems possess Finite \texorpdfstring{\(\mathcal{L}_2\)}{L2} Gain}
    Consider the case where each subsystem \(\bm{\mathcal{G}}_i\) for \( i \in \mathcal{N}\) has finite \(\mathcal{L}_2\) gain with \(\mbf{Q}_i = -\eye\), \(\mbf{S}_i = \mbf{0} \), and \(\mbf{R}_i = \gamma_i^2 \eye\), where \(\gamma_i \in \mathbb{R}_{>0}\). \Cref{theorem:2} suggests $\bar{\bm{\mathcal{G}}}$ is QSR-dissipative with
    \begin{align} \label{eqn:QSR_finite_L2_pre_scaling}
        \hspace{-7pt}
        \mbf{Q} &= -\frac{1}{\bar{\sigma}_{\mbs{\Psi}}^2}\eye,&
        \mbf{S} &= \mbf{0},&
        \mbf{R} &= \bar{\sigma}_{u} \max_{i \in \mathcal{N}}\gamma_i^2\eye,
    \end{align}
    provided the output scheduling matrices are active. Since \(\mbf{S} = \mbf{0}\), the pseudo-commutativity condition is trivially satisfied, effectively decoupling the input and output scheduling matrices. To get a more familiar form, \cref{eqn:QSR_finite_L2_pre_scaling} can be scaled by \(\bar{\sigma}_{\mbs{\Psi}}^2 > 0\) to obtain 
    \begin{align} \label{eqn:QSR_finite_L2_after_scaling}
        \mbf{Q} &= -\eye,&
        \mbf{S} &= \mbf{0},&
        \mbf{R} &= \bar{\sigma}_{\mbs{\Psi}}^2 \bar{\sigma}_{u} \max_{i \in \mathcal{N}} \gamma_i^2\eye.
    \end{align}
    Considering the base case and \Cref{col:QSR-Dissipative_of_GS_wrt_R} with \(\mathcal{R}_{>0} = \mathcal{N}\), substituting \cref{eqn:sigma_bar_u} and \cref{eqn:QSR_sigma_bar_Psi_for_scalar_si} into \cref{eqn:QSR_finite_L2_after_scaling} leads to
    \begin{align*}
        \mbf{Q} &= -\eye,&
        \mbf{S} &= \mbf{0},&
        \mbf{R} &= \gamma^2 \eye,
    \end{align*}
where $\bar{\bm{\mathcal{G}}}$ possess finite \(\mathcal{L}_2\) gain with the exact same \(\gamma\) derived in \cref{eqn:QSR_finite_L2_scalar_scheduling_gamma}, where the subsystems were assumed to be OSP\@. Moreover, the gain in \cref{eqn:QSR_finite_L2_scalar_scheduling_gamma} is similar to that of~\mbox{\cite[Theorem~5.2]{Forbes_Damaren}}. It is unclear which gain is more conservative since here, the maximum \(\gamma_i\) over all subsystems is used, whereas in~\mbox{\cite[Theorem~5.2]{Forbes_Damaren}}, the \(\gamma_i\) terms remain in the summation. On the other hand, \cref{eqn:QSR_finite_L2_scalar_scheduling_gamma} uses the supremum of sums of scheduling signals, \(s_i(t)\), whereas~\cite[Theorem~5.2]{Forbes_Damaren} uses sum of suprema of scheduling signals. Alternatively, \Cref{theorem:1} suggests $\bar{\bm{\mathcal{G}}}$ is QSR-dissipative with \(\mbf{Q} = -\eye\), \(\mbf{S} = \mbf{0}\), and  
    \begin{align} \label{eqn:QSR_finite_L2_using_theorem_1_general}
        \mbf{R} &= \mleft( N \sum_{i \in \mathcal{N}} \gamma_{i}^{2} \bar{\sigma}_{y,i}^2\bar{\sigma}_{u,i}^2 \mright)  \eye.
    \end{align}
    Again, for the base case, \cref{eqn:QSR_finite_L2_using_theorem_1_general} simplifies to
    \begin{align} \label{eqn:QSR_finite_L2_using_theorem_1_scalar}
        \mbf{R} &= \mleft( N \sum_{i \in \mathcal{N}} \gamma_{i}^{2} \mleft\| s_i \mright\|^{4}_{\infty}\mright) \eye,
    \end{align}
    where \(\mleft\| s_i \mright\|_{\infty} = \sup_{t \in \mathbb{R}_{\geq0}} \lvert s_i(t) \rvert\). Furthermore, applying~\mbox{\cite[Theorem~1]{QSR}} to the finite \(\mathcal{L}_2\) gain case recovers \cref{eqn:QSR_finite_L2_using_theorem_1_scalar}. Using the inequality in \Cref{lemma:AM-QM}, it can be shown that 
    \begin{equation*}
        N \sum_{i \in \mathcal{N}} \gamma_{i}^{2} \mleft\| s_i \mright\|_{\infty}^4 
    \geq \mleft( \sum_{i \in \mathcal{N}} \gamma_{i} \mleft\| s_i \mright\|_{\infty}^2\mright)^2,
    \end{equation*} 
    meaning,~\cite[Theorem~5.2]{Forbes_Damaren} is less conservative than \mbox{\Cref{theorem:1}} and~\cite[Theorem~1]{QSR} for the finite \(\mathcal{L}_2\) gain case.
\subsection{Subsystems are VSP}
    Consider the case where each subsystem \(\bm{\mathcal{G}}_i\) for \( i \in \mathcal{N}\) is VSP with \(\mbf{Q}_i = - \varepsilon_i\eye\), \(\mbf{S}_i = \frac{1}{2} \eye\), and \(\mbf{R}_i = - \delta_i\eye\), where \(\varepsilon_i, \delta_i \in \mathbb{R}_{>0}\). \Cref{theorem:2} suggests $\bar{\bm{\mathcal{G}}}$ is QSR-dissipative with 
        \begin{align*}
            \mbf{Q} &= - \varepsilon \eye,&
            \mbf{S} &= \frac{1}{2} \eye, &
            \mbf{R} &= - \delta \eye,
        \end{align*}
    with \(\varepsilon \in \mathbb{R}_{>0}\) and \(\delta \in \mathbb{R}_{\geq0}\) given by \cref{eqn:QSR_epsilon_OSP} and \cref{eqn:QSR_delta_ISP_not_strictly_active}, respectively, provided the output scheduling matrices are active and the scheduling matrices pseudo-commute with \(\mbf{S} = \frac{1}{2} \eye\). Furthermore, if the scheduling matrices are strongly active, \Cref{col:QSR-Dissipative_of_GS_wrt_R} with \(\mathcal{R}_{<0} = \mathcal{N}\) suggests $\bar{\bm{\mathcal{G}}}$ is VSP with \(\delta \in \mathbb{R}_{>0}\) given in \cref{eqn:QSR_delta_ISP_fully_active}.
    In~\cite{Forbes_Damaren}, the authors combine the ISP and finite \(\mathcal{L}_2\) gain results to show that the gain-scheduled system is VSP, while in~\cite{moalemi_forbes_ccta}, the ISP and OSP results are combined to achieve the same. Similar to the OSP case, \Cref{theorem:1} and~\cite[Theorem~1]{QSR} can still be used to show the overall gain-scheduled system $\bar{\bm{\mathcal{G}}}$ is QSR-dissipative. However, they cannot ensure $\bar{\bm{\mathcal{G}}}$ is VSP as they lead to \(\delta \in \mathbb{R}\) without the further guarantee of \(\delta > 0\).
\subsection{Subsystems are Conic}
    Consider the case where each subsystem \(\bm{\mathcal{G}}_i\) for \( i \in \mathcal{N}\) is conic with \(\mbf{Q}_i = -\eye\), \(\mbf{S}_i = c_i \eye\), and \(\mbf{R}_i = \mleft( r_i^2 - c_i^2 \mright) \eye\), where \(c_i \in \mathbb{R}\) and \(r_i \in \mathbb{R}_{>0} \), provided the scheduling matrices pseudo-commute with \(\mbf{S}_i\). \Cref{theorem:1} suggests $\bar{\bm{\mathcal{G}}}$ is QSR-dissipative with
    \begin{align} \label{eqn:QSR_conic_before_scaling}
            \mbf{Q} &= - \eye,&
            \mbf{S} &= c\eye,&
            \mbf{R} &= 
            \mleft(r^2-c^2\mright) \eye,
    \end{align}
    where
    \begin{subequations}
        \begin{align} 
            c &= \sum_{i \in \mathcal{N}}\mleft(c_i \bar{\sigma}_{y,i}^2\mright),\label{eqn:QSR_conic_center}
            \\%
            r &=%
            \sqrt{%
                N\sum_{i \in \mathcal{N}}
                \mleft( 
                    \delta_i
                    +
                    \bar{\sigma}_{y,i}^4 c_i^2
                \mright) 
            }%
            =%
            \sqrt{%
                N\sum_{i \in \mathcal{N}}
                \bar{r}_{i}^2
            },\label{eqn:QSR_conic_radius}
        \end{align}
    \end{subequations}
    are the conic center and radius, respectively, and \(\bar{r}_{i}^2 = \delta_i + \bar{\sigma}_{y,i}^4 c_i^2\). Following the definition of \(\delta_i\) in \cref{eqn:QSR_delta_i}, 
    \begin{align} \label{eqn:QSR_conic_r_bar}
        \bar{r}_{i}
        =
        \begin{cases}
            \bar{\sigma}_{y,i} \sqrt{\bar{\sigma}_{u,i}^2 \mleft( r_i^2 - c_i^2 \mright) + \bar{\sigma}_{y,i}^2 c_i^2}, 
            & \text{if } r_i > \lvert c_i \rvert , 
            \\
            \bar{\sigma}_{y,i} \sqrt{\bar{\nu}_{u,i}^2 \mleft( r_i^2 - c_i^2 \mright) +\bar{\sigma}_{y,i}^2 c_i^2},    
            & \text{if } r_i \leq \lvert c_i \rvert.
        \end{cases}
    \end{align}
    Additionally, the scheduling matrices pseudo-commute with \(\mbf{S}_i = c_i \eye\), meaning \(\mbs{\Phi}_{y, i}^{\trans}(t) \mbf{S}_i = \mbf{S}_i \mbs{\Phi}_{u, i}(t)\), for all \mbox{\(t \in \mathbb{R}_{\geq0}\)} and \(i \in \mathcal{N}\). It follows that
    \begin{equation*}
        \mbs{\Phi}_{y, i}^{\trans}(t) = \mbf{S}_i \mbs{\Phi}_{u, i}(t) \mbf{S}_i^{-1} = \mbs{\Phi}_{u, i}(t),
    \end{equation*}
    for all \(t \in \mathbb{R}_{\geq0}\) and \(i \in \mathcal{N}\), such that \(c_i \neq 0\). Therefore, the scheduling matrices corresponding to the \(i^{\text{th}}\) subsystem must share the same eigenvalues and singular values, provided \(c_i \neq 0\). Consequently, \cref{eqn:QSR_conic_r_bar} can be simplified as
    \begin{align} \label{eqn:QSR_conic_r_bar_2}
        \bar{r}_{i}
        =
        \begin{cases}
            \bar{\sigma}_{y,i}\bar{\sigma}_{u,i}r_i, 
            & \text{if } r_i > \lvert c_i \rvert = 0, 
            \\
            \bar{\sigma}_{i}^2 r_i, 
            & \text{if } r_i > \lvert c_i \rvert \neq 0, 
            \\
            \bar{\sigma}_{i}
            \sqrt{\bar{\nu}_{i}^2 r_i^2 + \mleft( \bar{\sigma}_{i}^2 - \bar{\nu}_{i}^2 \mright)c_i^2},    
            & \text{if } r_i \leq \lvert c_i \rvert,
        \end{cases}
    \end{align}
    where \(\bar{\sigma}_{i} = \bar{\sigma}_{u, i} = \bar{\sigma}_{y, i}\) and \(\bar{\nu}_{i} = \bar{\nu}_{u, i} = \bar{\nu}_{y, i}\) for all \(i \in \mathcal{N}\) such that \(c_i \neq 0\).
    For the base case, the conic center defined per~\cref{eqn:QSR_conic_center} simplifies to
        \(c = \sum_{i \in \mathcal{N}}c_i \mleft\| s_i \mright\|_{\infty}^{2}\),
   matching the results in~\cite[Theorem~1]{QSR} and~\cite[Theorem~6.2.1]{Forbes_thesis}. Moreover, for the base case, \(\bar{r}_i\) in \cref{eqn:QSR_conic_r_bar_2} simplifies to the results in~\cite[Theorem~1]{QSR} as
    \begin{align*} 
        \bar{r}_{i}
        =
        \begin{cases}
            \sigma_{i}^2 r_i, 
            & \text{if } r_i > \lvert c_i \rvert, 
            \\
            \sigma_{i}
            \sqrt{\nu_{i}^2 r_i^2 + \mleft( \sigma_{i}^2 - \nu_{i}^2 \mright)c_i^2},    
            & \text{if } r_i \leq \lvert c_i \rvert,
        \end{cases}
    \end{align*}
    with \(\sigma_i =  \mleft\| s_i \mright\|_{\infty} < \infty\) and \(\nu_{i} = \inf_{t \in \mathbb{R}_{\geq 0} } \lvert s_i(t) \rvert \geq 0\), for all \(i \in \mathcal{N}\). As discussed in~\cite{QSR}, when \(r_i > \lvert c_i \rvert \), \mbox{\cite[Theorem~6.2.1]{Forbes_thesis}} provides a less conservative bound on the conic radius of the gain-scheduled system compared to \Cref{theorem:1}. This is again as a direct result of the special case of Cauchy\textendash{}Schwartz inequality used in the proof of \Cref{theorem:1}. Alternatively, since \(\nu_{i} \leq \sigma_{i}\) for all \(i \in \mathcal{N}\), \Cref{theorem:1} provides a tighter bound on the conic radius of the gain-scheduled system when \(r_i \leq \lvert c_i \rvert \) compared to~\cite[Theorem~6.2.1]{Forbes_thesis}.

    \section{Application Example} \label{sec:simulation}
\subsection{QSR-Dissipative Properties of the Plant}
\begin{figure}[t]
    \centering
    \vspace{-3pt}
    \includegraphics{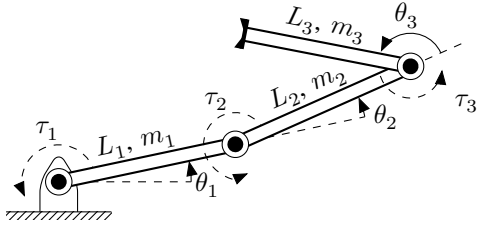}
    \vspace{-5pt}
    \caption{Planar rigid three-link robotic manipulator with a fixed base, damped joints, no forces acting on the end-effector, and no potential energy due to gravity}
    \label{fig:robot}
\end{figure}
\begin{figure}[t]
    \vspace{1pt}
    \centering
    \includegraphics{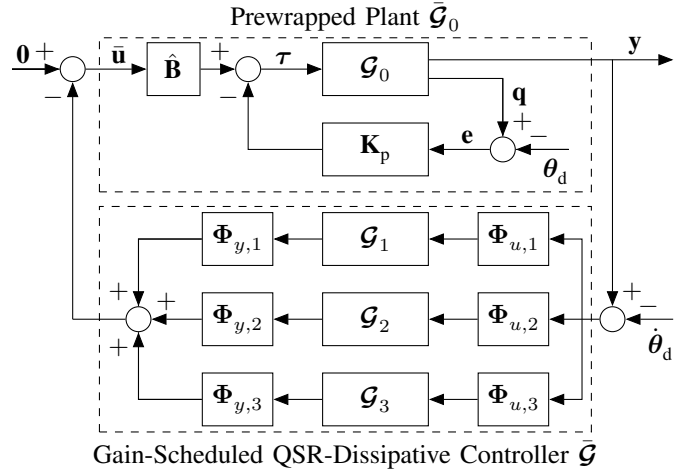}
    \caption{Gain-scheduled feedback control of the plant to be controlled. The input and output of each subcontroller are scheduled as per \cref{eq:GS_QSR_io} via the matrix multiplication between the scheduling matrices \(\mbs{\Phi}_{u, {\it i}}(t)\) and \(\mbs{\Phi}_{y, {\it i}}(t)\), and their corresponding signals \(\mbf{u}(t)\) and \(\mbf{y}_{{\it i}}(t)\), for \mbox{\({\it i} \in \mleft\{ 1, 2, 3 \mright\}\)}.}
    \label{fig:GS_feedback_control_system}
    \vspace{-5pt}
\end{figure}
    Consider a planar rigid three-link robotic manipulator, shown in \Cref{fig:robot}, with a fixed base, damped joints, no forces acting on the end-effector, and no potential energy due to gravity. The equations of motion of the robot are given by
    \begin{equation} \label{eq:robot_dynamics}
        \mbf{M}(\mbf{q}(t))\ddot{\mbf{q}}(t) + \mbf{D} \dot{\mbf{q}}(t) = 
        \mbf{f}_{\text{non}}(\mbf{q}(t), \dot{\mbf{q}}(t))
        +
        \mbs{\tau}(t), 
    \end{equation}
    where \(\mbf{M}(\mbf{q}(t)) = \mbf{M}^{\trans}(\mbf{q}(t)) \in \mathbb{S}_{++}^{3}\) is the mass matrix, \mbox{\(\mbf{D} = \diag \mleft(d_1, d_2, d_3\mright) \in \mathbb{S}_{++}^{3}\)} is the damping matrix, \(\mbf{f}_{\text{non}}(\mbf{q}(t), \dot{\mbf{q}}(t))\) captures the nonlinear inertial and Coriolis forces, \(\mbs{\tau}(t) = \begin{bmatrix} \tau_1(t) & \tau_2(t) & \tau_3(t)\end{bmatrix}^{\trans}\) are the joint torques, and \(\mbf{q}(t) = \begin{bmatrix} \theta_1(t) & \theta_2(t) & \theta_3(t)\end{bmatrix}^{\trans}\) are the generalized coordinates. The mass matrix and nonlinear forces can be obtained using the general Lagrangian dynamics formulation in~\cite[Section 8.1.2]{Modern_Robotics} as shown in~\cite{3DOF_robot}.
    A proportional control prewrap is added in negative feedback with the plant, \(\bm{\mathcal{G}}_0\), as per \Cref{fig:GS_feedback_control_system}, where \(\mbf{K}_\text{p} = \diag\mleft( k_{\text{p}, 1}, k_{\text{p}, 2}, k_{\text{p}, 3} \mright) \in \mathbb{S}_{++}^{3} \) and \(\mbs{\theta}_\text{d}(t) = \begin{bmatrix} \theta_{\text{d}, 1}(t) & \theta_{\text{d}, 2}(t) & \theta_{\text{d}, 3}(t)\end{bmatrix}^{\trans}\) is the desired trajectory. The resulting inner-loop, \(\bar{\bm{\mathcal{G}}}_0\), controls the displacement of the system. Dropping the time dependency for brevity, the nonlinear dynamics of the prewrapped system can be written as 
    \begin{equation} \label{eq:robot_dynamics_prewrap}
        \mbf{M}\ddot{\mbf{q}} = \mbf{f}_{\text{non}} - \mbf{D} \dot{\mbf{q}} + \hat{\mbf{B}}\bar{\mbf{u}} - \mbf{K}_\text{p} \mbf{e}
    \end{equation}
    where
    \(\hat{\mbf{B}} = 
        \begin{bmatrix}
            0 & 1 & 0 \\%
            0 & 0 & 1
        \end{bmatrix}^{\trans}\),
    \(\bar{\mbf{u}} = \begin{bmatrix} \bar{u}_1(t) & \bar{u}_2(t)\end{bmatrix}^{\trans}\) is from the output of the gain-scheduled controller \(\bar{\bm{\mathcal{G}}}\), and \(\mbf{e}(t) = \mbf{q}(t) - \mbs{\theta}_\text{d}(t)\) is the tracking error. The zero row in \(\hat{\mbf{B}}\) is purposely chosen to allow for gain-scheduling of non-square controllers. To analyze the dissipative property of the inner-loop, define the storage function
    \(
        V = 
        \frac{1}{2}\dot{\mbf{q}}^{\trans}\mbf{M} \dot{\mbf{q}}
        + 
        \frac{1}{2}\mbf{e}^{\trans} \mbf{K}_\text{p} \mbf{e}.
    \)
    It follows that
    \begin{equation} \label{eq:storage_function_before_simplification}
        \dot{V}
        = 
        \dot{\mbf{q}}^{\trans}
        \mleft( 
            \mbf{M} \ddot{\mbf{q}}
            +
            \frac{1}{2} \dot{\mbf{M}} \dot{\mbf{q}}
        \mright) 
        + 
        \dot{\mbf{e}}^{\trans} \mbf{K}_\text{p} \mbf{e}.
    \end{equation}
    For a constant tracking trajectory, it follows that \(\dot{\mbf{e}} = \dot{\mbf{q}}\). Substituting \cref{eq:robot_dynamics_prewrap} into \cref{eq:storage_function_before_simplification} and simplifying yields
    \begin{align}\label{eqn:storage_function_before_simplification_after_substitution}
        \dot{V}
        &=
        \dot{\mbf{q}}^{\trans}
        \mleft( 
            \mbf{f}_{\text{non}}
            +
            \frac{1}{2} \dot{\mbf{M}} \dot{\mbf{q}}
        \mright)
        - \dot{\mbf{q}}^{\trans}\mbf{D}\dot{\mbf{q}}
        +
        \dot{\mbf{q}}^{\trans}\hat{\mbf{B}}\bar{\mbf{u}}.
    \end{align}
    From the Lagrangian dynamics formulation of the system, it can be verified that \(\dot{\mbf{q}}^{\trans}
    \mleft( 
        \mbf{f}_{\text{non}} + \frac{1}{2} \dot{\mbf{M}} \dot{\mbf{q}}
    \mright) = 0\). Therefore, \cref{eqn:storage_function_before_simplification_after_substitution} simplifies to \( \dot{V} = - \dot{\mbf{q}}^{\trans}\mbf{D}\dot{\mbf{q}} +  \dot{\mbf{q}}^{\trans}\hat{\mbf{B}}\bar{\mbf{u}}\). Integrating both sides with \(t \in \left[ 0, T \right]\) yields
    \begin{align} \label{eq:storage_function}
        \hspace{-3pt}
        V(T) - V(0) &= \int_{0}^{T}  - \dot{\mbf{q}}^{\trans}(t)\mbf{D}\dot{\mbf{q}}(t) + \dot{\mbf{q}}^{\trans}(t) \hat{\mbf{B}}\bar{\mbf{u}}(t) \, \dt \nonumber
        \\%
        &= \left\langle \mbf{y}, -\mbf{D}\mbf{y} \right\rangle_{T} + \left\langle \mbf{y}, \hat{\mbf{B}}\bar{\mbf{u}} \right\rangle_{T},\quad \forall T \in \mathbb{R}_{\geq0},
    \end{align}
    where \(\mbf{y}(t) = \dot{\mbf{q}}(t) = \begin{bmatrix} \dot{\theta}_1(t) & \dot{\theta}_2(t) & \dot{\theta}_3(t)\end{bmatrix}^{\trans}\) is the measurement.
    Consequently, the prewrapped system \(\bar{\bm{\mathcal{G}}}_0\), in \mbox{\Cref{fig:GS_feedback_control_system}}, with input \(\bar{\mbf{u}}\mleft( t \mright) \) and output \(\mbf{y}\mleft( t \mright) \) is QSR-dissipative with \mbox{\(\mbf{Q}_\text{P} = -\mbf{D}\)}, \(\mbf{S}_\text{P} = \frac{1}{2}\hat{\mbf{B}}\), and \(\mbf{R}_\text{P} = \mbf{0}\). More precisely, as per~\mbox{\cite[Definition~3.1.2]{van_der_schaft}}, the prewrapped system is conservative with respect to the quadratic storage function constructed using \(\mbf{Q}_\text{P}\), \(\mbf{S}_\text{P}\), and \(\mbf{R}_\text{P}\).
    \subsection{Trajectory Generation}
    \begin{table}[t]
        \centering
        \caption{Discrete Joint Angles for Trajectory Generation}
        \vspace{-5pt}
        \label{tab:trajectory_points}
        \scriptsize
        \begin{tabularx}{\linewidth}{*{2}{>{\centering\arraybackslash}X}}
            \hline
            \hline
            \addlinespace[2pt]
            {\footnotesize Discrete Time Point \(t_k\)} & {\footnotesize Desired Joint Angle \( \bm{\theta}_\text{d}(t_{k})\)} \\
            {\footnotesize $\si{[s]}$} & {\footnotesize $\si{[deg]}$} \\
            \addlinespace[1pt]
            \hline
            \addlinespace[2pt]
            \(t_0 = 0 \) & \(\begin{bmatrix} 0^\circ & 160^\circ & -90^\circ \end{bmatrix}^{\trans} \) \\
            \(t_1 = 2 \) & \(\begin{bmatrix} 0^\circ & 160^\circ & -90^\circ \end{bmatrix}^{\trans} \) \\
            \(t_2 = 3 \) & \(\begin{bmatrix} 0^\circ & \hspace{-1pt}\phantom{-}45^\circ & \phantom{-}45^\circ \end{bmatrix}^{\trans} \) \\
            \(t_3 = 7 \) & \(\begin{bmatrix} 0^\circ & \hspace{-1pt}\phantom{-}45^\circ & \phantom{-}45^\circ \end{bmatrix}^{\trans} \) \\
            \(t_4 = 9 \) & \(\begin{bmatrix} 0^\circ & -90^\circ           & 160^\circ \end{bmatrix}^{\trans} \) \\
            \addlinespace[2pt]
            \hline
            \hline
        \end{tabularx}
    \end{table}
    The control objective is to have the three-link robot track a position, \(\bm{\theta}_\text{d}(t)\), and rate, \(\dot{\bm{\theta}}_\text{d}(t)\), trajectory. This is achieved by choosing discrete joint angles \(\bm{\theta}_\text{d}(t_k)\) and \(\bm{\theta}_\text{d}(t_{k+1})\) at times \(t_k\) and \(t_{k+1}\), and interpolating between them as such
    \begin{subequations}\label{eqn:trajectory_interpolation}
        \begin{gather}
            \begin{align}
                \eta(t)
                &= \frac{t - t_k}{t_{k+1} - t_k},&
                p_5(t) 
                &= 6 \eta^5 - 15\eta^4 + 10\eta^3,
            \end{align}
            \\%
            \bm{\theta}_\text{d}(t) 
            = 
            p_5(t) \mleft(\bm{\theta}_\text{d}(t_{k+1}) - \bm{\theta}_\text{d}(t_{k})\mright) 
                +
                \bm{\theta}_\text{d}(t_{k}).
        \end{gather}
    \end{subequations}
    As shown in \Cref{tab:trajectory_points}, the desired discrete joint angles are chosen such that the base remains fixed at \(0^\circ\) while the second and third joint angles operate within \(\left[-90^\circ, 160^\circ\right]\).
\subsection{QSR-Dissipative Control Synthesis}\label{subsec:QSR-Dissipative_Control_Synthesis}
    The subcontrollers to be gain-scheduled will be linear non-square QSR-dissipative controllers, synthesized using the linearized model of the system. In~\cite{QSR}, similar to \(\mathcal{H}_2\)-conic synthesis in~\cite{ryan_conic_forbes}, \(\mathcal{H}_2\)-optimal controllers are first designed using the linearized model of the system. These controllers are then rendered QSR-dissipative by solving a semidefinite program (SDP) involving a modified version of the linear matrix inequality (LMI) found in \Cref{eqn:QSR_LMI_Lemma}. Finally, the asymptotic stability of the closed-loop system is guaranteed using \Cref{eqn:QSR_Stability_Theorem}. Herein, a similar synthesis approach to~\mbox{\cite{QSR, Benhabib, Damaren_passive_map, walsh}} is taken, where QSR-dissipative controllers are synthesized using the solution to the linear-quadratic regulator (LQR) problem in conjunction with \Cref{eqn:QSR_LMI_Lemma} and \Cref{eqn:QSR_Stability_Theorem}.
    
    \begin{table}[t]
        \caption{Three-Link Manipulator Properties}
        \vspace{-5pt}
        \label{tab:Manipulator_Properties}
        \scriptsize
        \begin{tabularx}{\linewidth}{l @{\hspace{10pt}} *{3}{S[table-format=1.4]}}
            \hline
            \hline
            \addlinespace[2pt]
            {\footnotesize Link Parameters}
            & 
            {\footnotesize Link 1}
            &
            {\footnotesize Link 2} 
            &
            {\footnotesize Link 3} \\
            \addlinespace[1pt] 
            \hline
            \addlinespace[2pt]
            Length {$\si{[\meter]}$}                          & {\(\,L_1       = 1.10\)} & {\(\,L_2       = 0.60\)} & {\(\,L_3       = 0.50\)} \\
            Measured Length {$\si{[\meter]}$}                 & {\(\,\bar{L}_1 = 1.21\)} & {\(\,\bar{L}_2 = 0.54\)} & {\(\,\bar{L}_3 = 0.55\)} \\
            Mass {$\si{[\kilo\gram]}$}                        & {\(m_1       = 2.00\)}   & {\(m_2       = 0.90\)}   & {\(m_3       = 0.30\)} \\
            Measured Mass {$\si{[\kilo\gram]}$}               & {\(\bar{m}_1 = 2.40\)}   & {\(\bar{m}_2 = 0.72\)}   & {\(\bar{m}_3 = 0.36\)} \\
            Damping Coefficient {$\si[per-mode=symbol]{[\newton\mkern-4mu\cdot\mkern-4mu\meter\mkern-4mu\cdot\mkern-4mu\second\per\radian]}$}  & {\(\,\,d_1       = 5.00\)}  & {\(\,\,d_2       = 2.50\)}  & {\(\,\,d_3       = 2.50\)} \\
            \addlinespace[1pt]
            \hline
            \hline
        \end{tabularx}
    \end{table}
    \begin{table}[t]
        \caption{Controller Design Parameters}
        \vspace{-5pt}
        \label{tab:Controller_Properties}
        \scriptsize
        \begin{tabularx}{\linewidth}{l>{\centering\arraybackslash}X>{\centering\arraybackslash}X}
            \hline
            \hline
            \addlinespace[2pt]
            {\footnotesize Properties}
            & 
            {\footnotesize Symbol}
            &
            {\footnotesize Value} \\
            \addlinespace[1pt]
            \hline
            \addlinespace[2pt]
            Proportional Gain               & {\(\mbf{K}_{\text{p}}\)} & {\(\diag \mleft(5, 35, 35 \mright)\) } \\
            \addlinespace[2pt]
            \hline
            \addlinespace[2pt]
            \multirow{2}{*}{LQR Weights} & {\(\mbf{Q}_\text{LQR}\)} & \multicolumn{1}{@{\hspace{2pt}}c}{\(\diag \mleft(15, 15, 15, 10, 10, 10 \mright)^{-2}\)} \\
                                            & {\(\mbf{R}_\text{LQR}\)} & {\hspace{15pt}\(\diag \mleft(25, 25 \mright)^{-2}\)} \\
            \addlinespace[2pt]
            \hline
            \hline
        \end{tabularx}
        \vspace{-5pt}
    \end{table}
    First, the LQR problem requires a linearized state-space representation of \cref{eq:robot_dynamics}. As shown in~\cite{3DOF_robot}, \(\mbf{M}(\mbf{q}(t))\) is nonlinear with respect to the second and third joint angles,~\mbox{\(\theta_2(t)\) and \(\theta_3(t)\)}. Therefore, to cover the range of possible joint angles during the desired trajectory in \Cref{tab:trajectory_points}, three linearization points are chosen as \(\bar{\mbf{q}}_{1} = \bm{\theta}_\text{d}(t_{0})\), \(\bar{\mbf{q}}_{2} = \bm{\theta}_\text{d}(t_{2})\), and \(\bar{\mbf{q}}_{3} = \bm{\theta}_\text{d}(t_{3})\). The linearization of the prewrapped model, \(\bar{\bm{\mathcal{G}}}_0\) in \Cref{fig:GS_feedback_control_system}, about \(\bar{\mbf{q}}_{i}\), for \(i \in \mathcal{N} = \mleft\{ 1, 2, 3 \mright\}\), is given by
    \begin{align} \label{eq:linearized_model}
        \delta \dot{\mbf{x}}(t)
        &= 
        \mbf{A}_{i} \delta \mbf{x}(t)
        +
        \mbf{B}_{i} \delta \mbf{u}(t), &
        \delta \mbf{y}(t) &= \mbf{C}_{i} \delta \mbf{x}(t),
    \end{align}
    with
    \vspace{-5pt}
    \begin{subequations} \label{eq:linearized_model_A_B_C}
            \begin{align}
                \mbf{A}_{i}
                &= 
                    \begin{bmatrix}
                        \mbf{0}                               & \eye \\%
                        -\bar{\mbf{M}}^{-1}(\bar{\mbf{q}}_{i})\mbf{K}_{\text{p}} & -\bar{\mbf{M}}^{-1}(\bar{\mbf{q}}_{i})\mbf{D}
                    \end{bmatrix}, 
                \\%
                \mbf{B}_{i} 
                &=
                    \begin{bmatrix}
                        \mbf{0}  \\%
                        \bar{\mbf{M}}^{-1}(\bar{\mbf{q}}_{i})
                    \end{bmatrix}, \quad \quad
                \mbf{C}_{i}
                = 
                    \begin{bmatrix} 
                        \mbf{0} & \eye 
                    \end{bmatrix},
            \end{align}
    \end{subequations}
    where \(\bar{\mbf{M}}\mleft( \bar{\mbf{q}}_i \mright) \) is the measured mass matrix obtained using the measured link lengths and masses in \Cref{tab:Manipulator_Properties}, \(\mbf{K}_{\text{p}}\) is the proportional gain matrix in \Cref{tab:Controller_Properties}, and \mbox{\(\delta \mbf{x}(t) = \begin{bmatrix}
        \delta \mbf{q}^{\trans}(t)  & \delta \dot{\mbf{q}}^{\trans}(t)
    \end{bmatrix}^{\trans}\)}.
    Furthermore, the LQR problem's state and input weight matrices, \(\mbf{Q}_\text{LQR}\) and \(\mbf{R}_\text{LQR}\), are chosen following Bryson's rule~\cite{brysons}, 
    and are provided in \Cref{tab:Controller_Properties}. Using the linearization of the prewrapped model in \cref{eq:linearized_model} with the LQR state and input weight matrices, the gain matrix \(\mbf{K}_{i}\) is computed for each linearization point by solving the algebraic Riccati equation (ARE)~\cite{brysons}. Inspired by the LQG controller design in~\cite{Benhabib}, the three subcontrollers to be gain-scheduled are designed to have the state-space form
    \begin{subequations}
        \begin{align}
            \dot{\mbf{x}}_i(t) &= \mleft(\mbf{A}_{i} - \mbf{B}_{i} \mbf{K}_{i} - \mbf{B}_{\text{c}, i} \mbf{C}_i \mright)  \mbf{x}_i(t) + \mbf{B}_{\text{c}, i}\mbf{y}_i(t),
            \\%
            \mbf{u}_i(t)       &= \mbf{K}_{i} \mbf{x}_i(t),
        \end{align}
    \end{subequations}
    for \(i \in \mathcal{N}\), with \(\mbf{A}_i, \mbf{B}_i\), and \(\mbf{C}_i\) in \cref{eq:linearized_model_A_B_C}, and \(\mbf{B}_{\text{c}, i}\) is to be constructed such that the subcontrollers are QSR-dissipative. To synthesize QSR-dissipative subcontrollers, an SDP is constructed using a modified version of the LMI in \Cref{eqn:QSR_LMI_Lemma} through a change of variable, \(\mbf{F}_i = \mbf{P}_i\mbf{B}_{\text{c}, i}\), and solved using \texttt{CVX}~\cite{CVXPY, CVXPY_paper}, and \texttt{MOSEK}~\cite{mosek}. This involves solving for \(\mbf{P}_i \in \mathbb{S}_{++}^{6}\), \(\mbf{Q}_{\text{c}, i} \in \mathbb{S}^{2}\), \(\mbf{R}_{\text{c}, i} \in \mathbb{S}^{3}\), \(\mbf{F}_i \in \mathbb{R}^{6 \times 3}\), and \(\rho_i \in \mathbb{R}_{>0}\) subject to
    \begin{subequations} \label{eqn:SDP}
        \begin{align} 
            \begin{bmatrix}
                \mbf{P}_i\hat{\mbf{A}}_i + \hat{\mbf{A}}_i^{\trans} \mbf{P}_i - \mbf{F}_i\mbf{C}_i - \mbf{C}_{i}^{\trans}\mbf{F}_{i}^{\trans} - \hat{\mbf{Q}}_i & \mbf{F}_i - \hat{\mbf{S}}_i \\%
                \mbf{F}_{i}^{\trans}  - \hat{\mbf{S}}_{i}^{\trans}      & -\mbf{R}_{\text{c}, i}
            \end{bmatrix}
            &\preceq 0, \label{eqn:QSR_LMI_modified_1}
            \\%
            \begin{bmatrix}
                \rho_i\mbf{Q}_{\text{P}} + \mbf{R}_{\text{c}, i} & - \rho_i\mbf{S}_{\text{P}} + \mbf{S}_{\text{c}, i}^{\trans}  \\%
                - \rho_i\mbf{S}_{\text{P}}^{\trans} + \mbf{S}_{\text{c}, i}   & \mbf{Q}_{\text{c}, i}
            \end{bmatrix}
            &\prec 0, \label{eqn:QSR_LMI_modified_2}
        \end{align}
    \end{subequations}
    where
    \begin{align} \label{eqn:SDP_change_of_variables}
        \hat{\mbf{A}}_i &= \mbf{A}_{i} - \mbf{B}_{i}\mbf{K}_{i}, & \hat{\mbf{Q}}_i &= \mbf{K}_i^{\trans} \mbf{Q}_{\text{c}, i} \mbf{K}_i, & \hat{\mbf{S}}_i &= \mbf{K}_i^{\trans} \mbf{S}_{\text{c}, i},
    \end{align}
    and \(\mbf{S}_{\text{c}, i} = \mbf{S}_{\text{P}}^{\trans} \) for \(i \in \mathcal{N}\). Additionally, \(\mbf{B}_{\text{c}, i}\) is obtained from \(\mbf{F}_i\) as \(\mbf{B}_{\text{c}, i} = \mbf{P}_i^{-1}\mbf{F}_i\).
    
    The LMI in \cref{eqn:QSR_LMI_modified_1} is obtained from \Cref{eqn:QSR_LMI_Lemma} and ensures the subcontrollers are QSR-dissipative. Additionally, in the absence of gain-scheduling, the LMI in \cref{eqn:QSR_LMI_modified_2} guarantees the asymptotic stability of the negative feedback interconnection of any individual subcontroller with the prewrapped plant \(\bar{\bm{\mathcal{G}}}_0\) as per \Cref{eqn:QSR_Stability_Theorem}. Therefore, the LMI in \cref{eqn:SDP} is a sufficient condition for the synthesis of QSR-dissipative subcontrollers that render the closed-loop system asymptotically stable when connected in negative feedback with the prewrapped plant.

    Similar to the control synthesis techniques in~\cite{ryan_conic_forbes, QSR}, a cost function can be used in tandem with the LMI in \cref{eqn:SDP} to impose certain design constraints.
    Here, from \cref{eqn:SDP}, it can be shown using the Schur complement lemma that \mbox{\(\mbf{R}_{\text{c}, i} \in \mathbb{S}_{+}^{3}\)} and \(\mbf{Q}_{\text{c}, i} \in \mathbb{S}_{--}^{2}\). Therefore, minimizing the cost function \(\mathcal{J}\mleft( \mbf{R}_{\text{c}, i} \mright) = \trace \mleft( \mbf{R}_{\text{c}, i} \mright)\) results in \(\mbf{R}_{\text{c}, i} = \mbf{0}\) and subsequently \(\mbf{F}_i = \hat{\mbf{S}}_i\). This choice of cost function is motivated by the fact that solving the SDP in \cref{eqn:SDP} subject to \(\mbf{R}_{\text{c}, i} = \mbf{0}\) and \(\rho_i = 1\) is equivalent to solving for \(\mbf{P}_i \in \mathbb{S}_{++}^{6}\) and \(\mbf{Q}_{\text{c}, i} \in \mathbb{S}_{--}^{2}\), subject to
    \begin{equation} \label{eqn:SDP_with_R_eq_0}
        \mbf{P}_i\hat{\mbf{A}}_i + \hat{\mbf{A}}_i^{\trans} \mbf{P}_i - \hat{\mbf{S}}_i\mbf{C}_i - \mbf{C}_i^{\trans}\hat{\mbf{S}}_i^{\trans} - \hat{\mbf{Q}}_i \preceq 0,
    \end{equation}
    with \(\hat{\mbf{A}}_i\), \(\hat{\mbf{Q}}_i\), and \(\hat{\mbf{S}}_i\) found in \cref{eqn:SDP_change_of_variables}. Again, \(\mbf{B}_{\text{c}, i}\) can be obtained from \(\mbf{F}_i = \hat{\mbf{S}}_i = \mbf{P}_i\mbf{B}_{\text{c}, i}\) as \(\mbf{B}_{\text{c}, i} = \mbf{P}_{i}^{-1} \hat{\mbf{S}}_i\). Compared to~\cref{eqn:SDP}, the SDP involved in solving \cref{eqn:SDP_with_R_eq_0} is computationally less expensive and results in a subcontroller \(\bm{\mathcal{G}}_i\) that is QSR-dissipative with \(\mbf{Q}_{\text{c}, i} \in \mathbb{S}_{--}^{2}\), \(\mbf{S}_{\text{c}, i} = \frac{1}{2}\hat{\mbf{B}}^{\trans}\), and \(\mbf{R}_{\text{c}, i} = \mbf{0}\). Three QSR-dissipative subcontrollers, \(\bm{\mathcal{G}}_i\) for \(i \in \mathcal{N}\), are synthesized using the linearization of the prewrapped model in \cref{eq:linearized_model} about the linearization points \(\bar{\mbf{q}}_i\) and the LQR weights in \Cref{tab:Controller_Properties}. The subcontrollers are then gain-scheduled to form the overall gain-scheduled controller \(\bar{\bm{\mathcal{G}}}\) in \Cref{fig:GS_feedback_control_system}.
    Since \(\mbf{Q}_{\text{c}, i} \in \mathbb{S}_{--}^{2}\) and the subcontrollers share a common \(\mbf{S}_{\text{c}, i}\) for \(i \in \mathcal{N}\), \Cref{theorem:2} suggest the overall gain-scheduled controller $\bar{\bm{\mathcal{G}}}$, in \Cref{fig:GS_feedback_control_system}, is QSR-dissipative, with \(\mbf{Q}_{\text{c}} = -\varepsilon \eye  \in \mathbb{S}_{--}^{2}\), \(\mbf{S}_{\text{c}} = \mbf{S}_{\text{c}}^{\trans} = \frac{1}{2}\hat{\mbf{B}}^{\trans}\), and \(\mbf{R}_{\text{c}} = \mbf{0} \) with \(\varepsilon\) defined in \Cref{eqn:Q_epsilon}, provided the output scheduling matrices are active, and the scheduling matrices pseudo-commute with \(\hat{\mbf{B}}^{\trans}\). Finally, the closed-loop system in \Cref{fig:GS_feedback_control_system} is guaranteed to be asymptotically stable, since the LMI in \cref{eqn:QSR_Stability_Theorem_LMI} can be trivially satisfied with \(\rho = 1\).

\subsection{Scheduling Function Construction} \label{subsec:scheduling_signals}
    Similar to~\cite{moalemi_forbes_ccta}, fourth degree scalar polynomials are used as scheduling signals in the construction of scheduling matrices. For the linearization points \(\bar{\mbf{q}}_{i}\), the scheduling signals in \Cref{fig:signals} are defined as
    \begin{subequations} \label{eq:scalar_scheduling_signals}
        \begin{align}
            s_1(t) &= 
            \begin{cases}
                1                                  & \qquad \phantom{-} 0.0 \leq t < 1.0, \\
                1 - \left(\frac{t-1}{3}\right)^4   & \qquad \phantom{-} 1.0 \leq t \leq 4.0, \\
                0                                  & \qquad \phantom{-} 4.0 < t,
            \end{cases} \\
            s_2(t) &= 
            \begin{cases}
                1 - \left(\frac{t-5}{4}\right)^4   & \qquad \phantom{-}  1.0 \leq t \leq 9.0 ,\\
                0                                  & \qquad \phantom{-} \text{otherwise},
            \end{cases} \\
            s_3(t) &= 
            \begin{cases}
                0                                  & \qquad \phantom{-} 0.0 \leq t < 7.0, \\
                1 - \left(\frac{t-9}{2}\right)^4   & \qquad \phantom{-} 7.0 \leq t \leq 9.0, \\
                1                                  & \qquad \phantom{-} 9.0 < t,
            \end{cases}
        \end{align}
    \end{subequations}
    where for $9 < t$, $s_3(t) = 1$, while the other signals are zero. Additionally, at all times, all scheduling signals are bounded and for each time \(t \in \mathbb{R}_{\geq0}\), at least one scheduling signal is active, meaning \(\forall t \in \mathbb{R}_{\geq0}, \exists i \in \mathcal{N}\) such that \(s_i(t) \neq 0\). Contrary to~\cite{QSR}, the scheduling signals in \cref{eq:scalar_scheduling_signals} are \mbox{``turned off''}, in the sense that \(0 = \nu_i \leq \lvert s_i(t) \rvert \leq 1\) for \(i \in \mathcal{N}\) and \(\forall t \in \mathbb{R}_{\geq0}\).
    \begin{figure}[t]
        \centering
        \vspace{5pt}
        \includegraphics[width=0.485\textwidth]{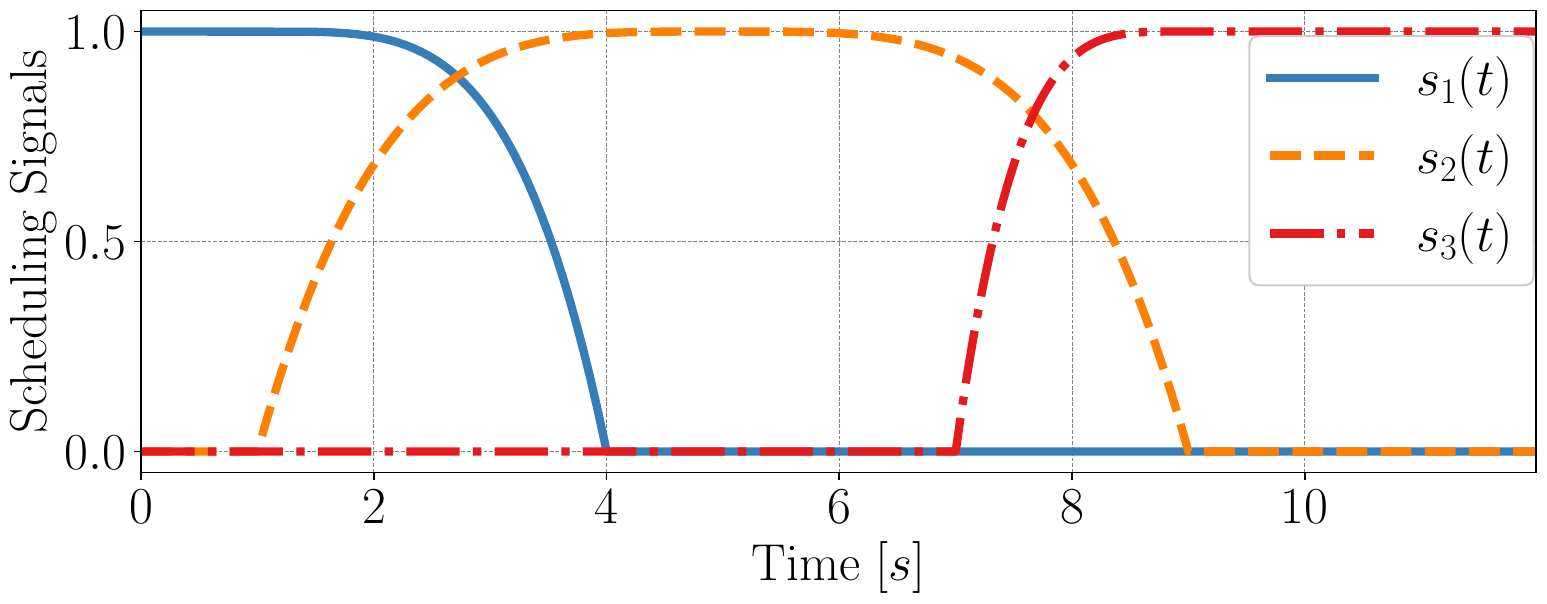}
        \vspace{-15pt}
        \caption{Scheduling signals $s_1(t)$, $s_2(t)$, and $s_3(t)$ defined in \cref{eq:scalar_scheduling_signals}.}
        \label{fig:signals}
        \vspace{-5pt}
    \end{figure}
    
    As discussed in \Cref{subsec:QSR-Dissipative_Control_Synthesis}, to guarantee the gain-scheduled controller $\bar{\bm{\mathcal{G}}}$ is QSR-dissipative, the scheduling matrices need to pseudo-commute with \(\mbf{S}_{\text{c}} = \mbf{S}_{\text{P}}^{\trans} = \frac{1}{2}\hat{\mbf{B}}^{\trans}\), while the output scheduling matrices need to be active. Furthermore, as discussed in \Cref{subsec:scheduling_matrix_construction}, the scheduling matrices can be design using \(\mbf{S}_{\text{c}}\) to satisfy this pseudo-commutativity condition. Consider the SVD of \(\mbf{S}_{\text{c}} = \frac{1}{2} \begin{bmatrix} \mbf{0} & \eye \end{bmatrix}\) given by 
    \begin{align*}
        \mbf{U} &=  
        \begin{bmatrix} 
            1 & 0\\%
            0 & 1
        \end{bmatrix}, &
         \mbs{\Sigma}_1 &= \frac{1}{2}         
         \begin{bmatrix} 
            1 & 0\\%
            0 & 1
        \end{bmatrix},& 
         \mbf{V} &=
        \begin{bmatrix} 
            0 & 0 & 1\\%
            1 & 0 & 0\\%
            0 & 1 & 0
        \end{bmatrix},
    \end{align*}
    such that \(\mbf{S}_{\text{c}} = \mbf{U}\begin{bmatrix} \mbs{\Sigma}_1 & \mbf{0} \end{bmatrix}\mbf{V}^{\trans}\). As per \Cref{lemma:AS_eq_SB}, to satisfy the pseudo-commutativity condition, \(\mbs{\Phi}_{y, i}^{\trans}(t) \mbf{S}_{\text{c}} = \mbf{S}_{\text{c}} \mbs{\Phi}_{u, i}(t)\), the scheduling matrices need to satisfy
    \begin{subequations}\label{eqn:scheduling_matrices_after_svd}
        \begin{align}
            \mbs{\Phi}_{u, i}(t) &= 
            \mbf{V} 
            \begin{bmatrix}
                \bar{\mbf{Z}}_{11, i}(t) & \mbf{0} \\%
                \bar{\mbf{z}}_{21, i}(t) & \bar{z}_{i}(t)
            \end{bmatrix}
            \mbf{V}^{\trans}
            =
            \mbf{V} 
            \bar{\mbf{Z}}_{i}(t)
            \mbf{V}^{\trans},
            \\%
            \mbs{\Phi}_{y, i}(t) &= \mbf{U} \bar{\mbf{Z}}_{11, i}^{\trans}(t) \mbf{U}^{\trans} = \bar{\mbf{Z}}_{11, i}^{\trans}(t),
        \end{align}
    \end{subequations}
    for all \(t \in \mathbb{R}_{\geq0}\), where \(\bar{\mbf{Z}}_{11, i}(t) \in \mathbb{R}^{2\times 2}\), \(\bar{\mbf{z}}_{21, i}(t) \in \mathbb{R}^{1 \times 2}\), and \(\bar{z}_{i}(t) \in \mathbb{R}\), for \(i \in \mathcal{N}\), are design variables. Therefore, scheduling of each subcontroller \(\bm{\mathcal{G}}_i\) in \Cref{fig:GS_feedback_control_system} requires seven hyperparameters. Herein, \(\bar{\mbf{Z}}_{i}(t)\) for each subcontroller is chosen as
    \begin{subequations} \label{eqn:matrix_scheduling_signals}
        \begin{align}
            \bar{\mbf{Z}}_{1}(t) &= 
                \left[
                    \begin{array}{cc;{1pt/1pt}c}
                        s_1(t) & -0.5 s_1(t) & 0 \\
                        0      & 0           & 0 \\ \hdashline[1pt/1pt]
                        s_3(t) & s_2(t)      & s_1(t)
                    \end{array}
                \right],\\%
                \bar{\mbf{Z}}_{2}(t) &= 
                \left[
                    \begin{array}{cc;{1pt/1pt}c}
                        s_1(t) + s_2(t) & 0               & 0\\ 
                        0               & s_2(t) + s_3(t) & 0\\ \hdashline[1pt/1pt]
                        0               & 0               & s_2(t)\\
                    \end{array}
                \right],\\%
                \bar{\mbf{Z}}_{3}(t) &= 
                \left[
                    \begin{array}{cc;{1pt/1pt}c}
                        s_3(t) & 0      & 0\\ 
                        s_1(t) & s_3(t) & 0\\ \hdashline[1pt/1pt]
                        0      & 0      & s_2(t) + s_3(t)\\
                    \end{array}
                \right],\
        \end{align} 
    \end{subequations}
    for all \(t \in \mathbb{R}_{\geq0}\), where \(s_1(t)\), \(s_2(t)\), and \(s_3(t)\) are defined in \cref{eq:scalar_scheduling_signals}. The scheduling matrices constructed as per \cref{eqn:scheduling_matrices_after_svd} with \(\bar{\mbf{Z}}_{i}(t)\) in \cref{eqn:matrix_scheduling_signals} are chosen to demonstrate the various possible structures. For instance, \(\bar{\mbf{Z}}_{1}(t)\) is rank deficient for all \(t \in \mathbb{R}_{\geq0}\), \(\bar{\mbf{Z}}_{2}(t)\) is diagonal, and \(\bar{\mbf{Z}}_{3}(t)\) is lower-triangular. Finally, as required by \Cref{theorem:2}, the output scheduling matrices in \cref{eqn:scheduling_matrices_after_svd} are active, that is, \(\forall t \in \mathbb{R}_{\geq0}, \exists i \in \mathcal{N} = \mleft\{ 1, 2, 3 \mright\}\) such that \( \mbs{\Phi}_{y, i}(t) = \bar{\mbf{Z}}_{11, i}^{\trans}(t) \neq \mbf{0}\).
\subsection{Comparison}
    A closed-loop simulation is performed using three different control approaches to have the three-link planar manipulator in \Cref{fig:robot} track the trajectory given by \cref{eqn:trajectory_interpolation}. As a baseline, the controller \(\bm{\mathcal{G}}_3\), a QSR-dissipative controller synthesized using the SDP in \Cref{subsec:QSR-Dissipative_Control_Synthesis}, is used and will be referred to as the unscheduled controller henceforth. The second approach, referred to as the scalar gain-scheduled (GS) controller, uses the scalar scheduling signals in \cref{eq:scalar_scheduling_signals} to gain-schedule the three subcontrollers, \(\bm{\mathcal{G}}_1\), \(\bm{\mathcal{G}}_2\), and \(\bm{\mathcal{G}}_3\), designed about the linearization points \(\bar{\mbf{q}}_1\), \(\bar{\mbf{q}}_2\), and \(\bar{\mbf{q}}_3\), respectively. This approach is inspired by the VSP gain-scheduling architecture in~\cite{Forbes_Damaren} but with non-square subcontrollers gain-scheduled as per \Cref{fig:GS_feedback_control_system}, where \(\mbs{\Phi}_{u, i}(t) =  s_i(t)\eye\) and \(\mbs{\Phi}_{y, i}(t) = s_i(t)\eye\) for all \(i \in \mathcal{N}\). The third approach, referred to as the matrix GS controller, uses the scheduling matrices constructed as per \cref{eqn:scheduling_matrices_after_svd,eqn:matrix_scheduling_signals} instead. Finally, since the controllers designed in~\mbox{\Cref{subsec:QSR-Dissipative_Control_Synthesis}} are meant to control the second and third joint angles, the comparison will only focus on these joints. The desired trajectory generated using \cref{eqn:trajectory_interpolation} with the discrete points in \Cref{tab:trajectory_points} is shown in \Cref{fig:trajectory_comparison} where the close tracking performance of the three control approaches is shown. The joint angle error, \(\mbf{e}(t) = \begin{bmatrix} e_1(t) & e_2(t) & e_3(t) \end{bmatrix}^{\trans} = \mbf{q}(t) - \bm{\theta}_\text{d}(t)\), is shown in \Cref{fig:error_comparison}, where the matrix GS controller has noticeably less error, magnitude wise, than the scalar GS controller when the three-link manipulator transitions between desired positions. The root-mean-square (RMS) joint angle error, and joint angle error rates are tabulated in \Cref{tab:RMS_error}. Again, the matrix GS controller realizes lower RMS angle error and RMS angle rate error. The code used to generate the figures presented in this section can be found in the GitHub repository at \url{https://github.com/decargroup/matrix_scheduling_qsr_dissipative_systems}.
    \begin{figure}
        \centering
        \vspace{4pt}
        \includegraphics[width=0.485\textwidth]{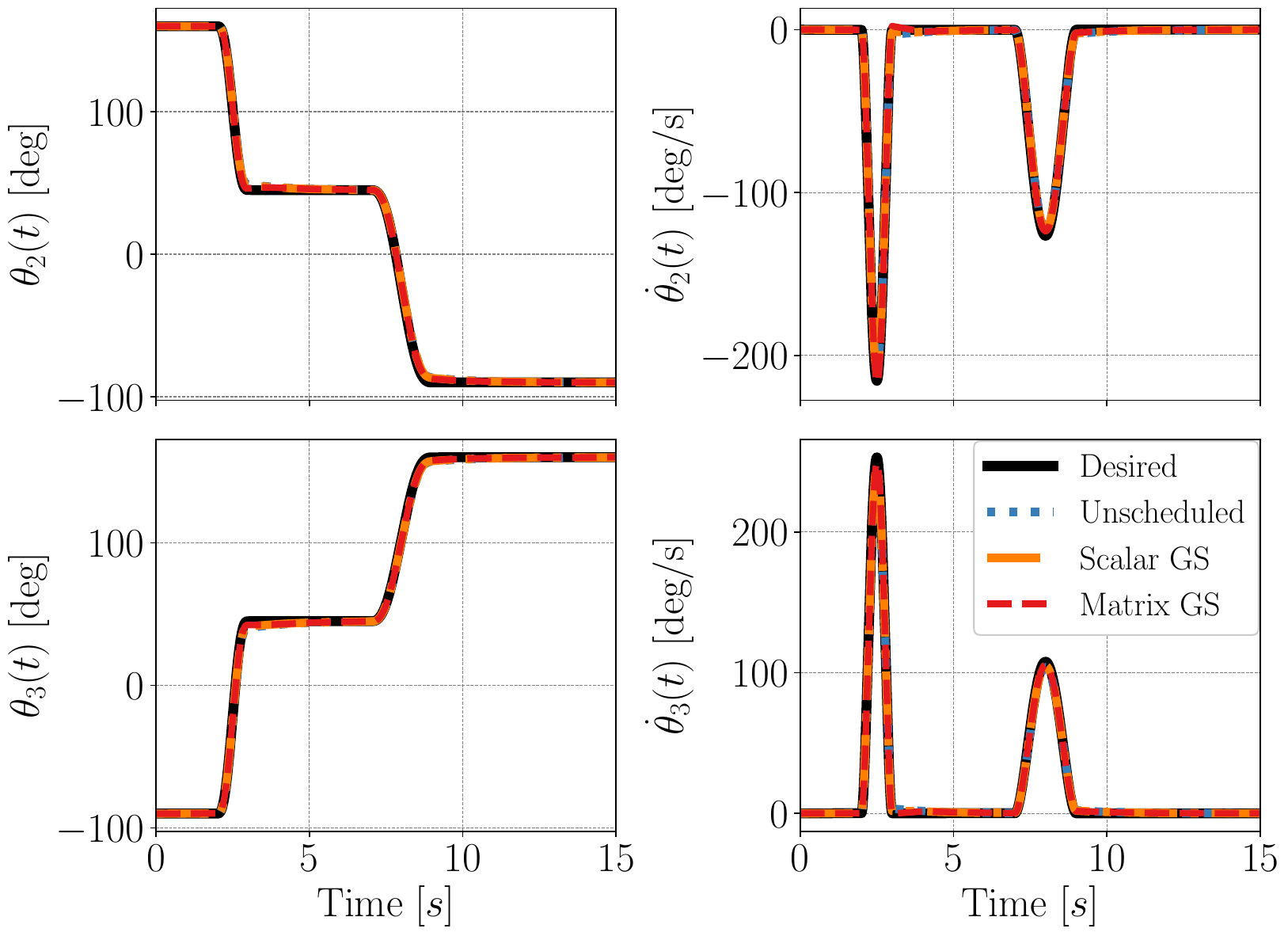}
        \captionof{figure}{Comparison of joint angles and joint rates.}
        \vspace{18pt}
        
        \label{fig:trajectory_comparison}
        \includegraphics[width=0.485\textwidth]{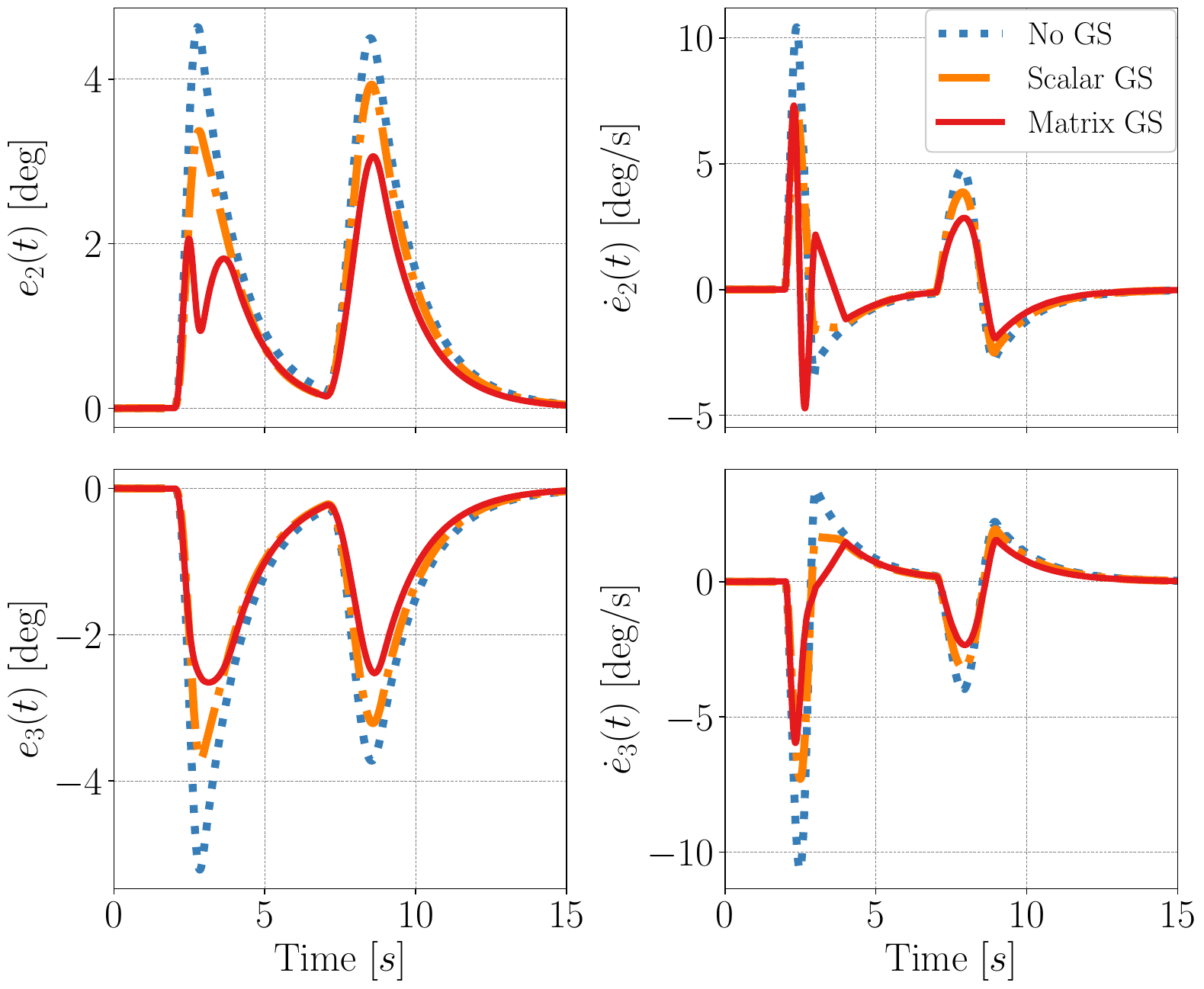}
        \captionof{figure}{Comparison of joint angles errors and error rates.}
        \vspace{18pt}

        \label{fig:error_comparison}
        \includegraphics[width=0.485\textwidth]{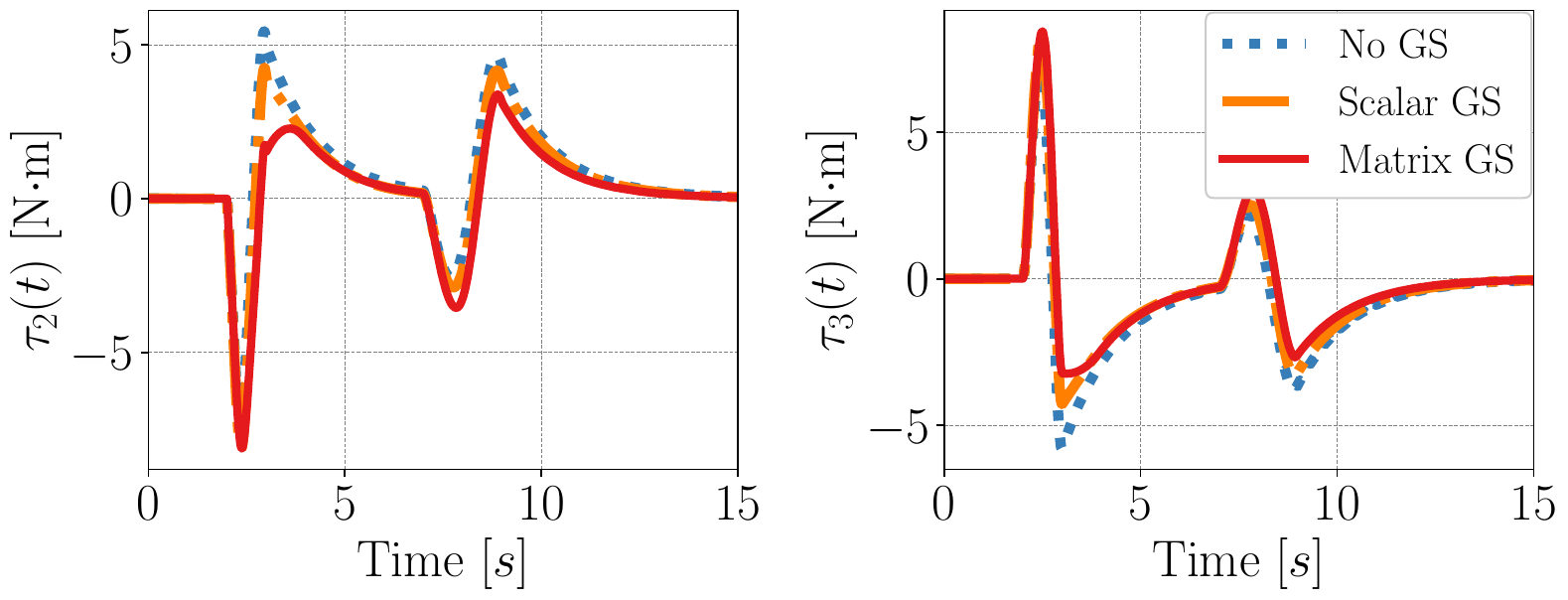}
        \captionof{figure}{Comparison of joint torques.}
        \vspace{18pt}

        \label{fig:torques_comparison}
        \captionof{table}{RMS Error of Joint Angle and Joint Angle Rate}
        \vspace{-5pt}
        \label{tab:RMS_error}
        \small
        \begin{tabularx}{\linewidth}{l *{2}{>{\centering\arraybackslash}X} *{2}{>{\centering\arraybackslash}X}}
            \hline
            \hline
            \addlinespace[2pt]
            & 
            \multicolumn{2}{c}{RMS angle error}
            &
            \multicolumn{2}{c}{RMS angle rate error} 
            \\
            \multirow[b]{2}{*}{\vspace{-8pt}Control method} 
            & 
            \multicolumn{2}{c}{$\si{[deg]}$}
            &
            \multicolumn{2}{c}{$\si{[deg/s]}$} 
            \\
            \cmidrule(l){2-3}
            \cmidrule(l){4-5}
            & {$e_2$} & {$e_3$} & {$\dot{e}_2$} & {$\dot{e}_3$} \\
            \hline
            \addlinespace[2pt]
            Unscheduled        & 1.8239          & 1.7892          & 2.1111          & 2.1191 \\
            Scalar scheduling  & 1.5226          & 1.4358          & 1.6073          & 1.5442 \\
            Matrix scheduling  & \textbf{1.1304} & \textbf{1.2015} & \textbf{1.3669} & \textbf{1.1356} \\
            \hline
            \hline
        \end{tabularx}
    \end{figure}
    \section{Closing Remarks}\label{sec:closing_remarks}
The gain-scheduling of QSR-dissipative subsystems using scheduling matrices is considered in this paper. The proposed gain-scheduling architecture is shown to result in an overall QSR-dissipative gain-scheduled system, provided the scheduling matrices are bounded and possess certain pseudo-commutative properties. A detailed discussion is provided regarding the design and construction of scheduling matrices that satisfy the introduced pseudo-commutative property. Furthermore, the gain-scheduling of a broader class of QSR-dissipative systems is considered than that found in prior work. For the various passive cases of QSR-dissipativity, the conditions on the scheduling matrices reduce to the same conditions on the scheduling signals reported in~\cite{Damaren_passive_map, Forbes_Damaren}, while for the conic case, the results of~\cite{QSR} were recovered when the scheduling matrices are chosen to be a scalar times the identity matrix. The proposed gain-scheduling architecture, along with the LMI-based QSR-dissipative control synthesis technique, is applied to the control of a planar three-link robotic manipulator subject to model uncertainty. Numerical simulation results highlight the added benefits of using scheduling matrices compared to scheduling signals.
    \section{Acknowledgment}
The authors would like to thank Prof. Xiao-Wen Chang for their helpful input with regard to \Cref{lemma:AS_eq_SB}.
    \appendix
\begin{lemma} \label{lemma:AS_eq_SB}
    Consider a real nonzero matrix \(\mbf{S} \in \mathbb{R}^{m \times n}\) such that \(\mbf{S} \neq \mbf{0}\). The arbitrary matrices \(\mbf{A} \in \mathbb{R}^{m \times m}\) and \(\mbf{B} \in \mathbb{R}^{n \times n}\) satisfy \(\mbf{A}\mbf{S} = \mbf{S}\mbf{B}\) if and only if
    \begin{align*}
        \mbf{A} &= 
        \mbf{U} 
        \begin{bmatrix}
            \mbs{\Sigma}_{1}\bar{\mbf{B}}_{11}\mbs{\Sigma}_{1}^{-1} & \bar{\mbf{A}}_{12} \\%
            \mbf{0}                                                 & \bar{\mbf{A}}_{22}
        \end{bmatrix}
        \mbf{U}^{\trans},
        &
        \mbf{B} &= 
        \mbf{V} 
        \begin{bmatrix}
            \bar{\mbf{B}}_{11} & \mbf{0} \\%
            \bar{\mbf{B}}_{21} & \bar{\mbf{B}}_{22}
        \end{bmatrix}
        \mbf{V}^{\trans},
    \end{align*}
    where \(\mbf{U} \in \mathbb{R}^{m \times m}\), \(\mbs{\Sigma}_{1} \in \mathbb{S}^{\varrho}\), and \(\mbf{V} \in \mathbb{R}^{n \times n}\) are given by the singular value decomposition (SVD) of \(\mbf{S}\) such that 
    \begin{equation*} 
        \mbf{S} = \mbf{U} 
        \begin{bmatrix}
            \mbs{\Sigma}_{1} & \mbf{0} \\%
            \mbf{0}          & \mbf{0}
        \end{bmatrix}
        \mbf{V}^{\trans},
    \end{equation*} 
    with \(\varrho = \rank\mleft( \mbf{S} \mright) \geq 1\), and \(\bar{\mbf{A}}_{12}, \bar{\mbf{A}}_{22}, \bar{\mbf{B}}_{11}, \bar{\mbf{B}}_{21}\), and \(\bar{\mbf{B}}_{22}\) are arbitrary real matrices of appropriate dimensions.
\end{lemma}
\begin{proof}
    Assume \(\mbf{A} \in \mathbb{R}^{m \times m}\) and \(\mbf{B} \in \mathbb{R}^{n \times n}\) satisfy \mbox{\(\mbf{A}\mbf{S} = \mbf{S}\mbf{B}\)}. Using the SVD of \(\mbf{S}\), it follows that
    \begin{equation} \label{eqn:AS_eq_SB_step_1}
        \mbf{A}\mbf{U} 
        \begin{bmatrix}
            \mbs{\Sigma}_{1} & \mbf{0} \\%
            \mbf{0}          & \mbf{0}
        \end{bmatrix}
        \mbf{V}^{\trans}
        =
        \mbf{U} 
        \begin{bmatrix}
            \mbs{\Sigma}_{1} & \mbf{0} \\%
            \mbf{0}          & \mbf{0}
        \end{bmatrix}
        \mbf{V}^{\trans}
        \mbf{B}.
    \end{equation}
    Since \(\mbf{U}\) and \(\mbf{V}\) are orthogonal matrices, left and right multiplying \cref{eqn:AS_eq_SB_step_1} by \(\mbf{U}^{\trans}\) and \(\mbf{V}\), respectively, yields
    \begin{equation} \label{eqn:AS_eq_SB_step_2}
        \mbf{U}^{\trans}\mbf{A}\mbf{U} 
        \begin{bmatrix}
            \mbs{\Sigma}_{1} & \mbf{0} \\%
            \mbf{0}          & \mbf{0}
        \end{bmatrix}
        =
        \begin{bmatrix}
            \mbs{\Sigma}_{1} & \mbf{0} \\%
            \mbf{0}          & \mbf{0}
        \end{bmatrix}
        \mbf{V}^{\trans}\mbf{B}\mbf{V}.
    \end{equation}
    Define
    \begin{align} \label{eqn:AS_eq_SB_A_bar_B_bar} 
        \mbf{U}^{\trans}\mbf{A}\mbf{U} 
        &= 
        \begin{bmatrix}
            \bar{\mbf{A}}_{11} & \bar{\mbf{A}}_{12} \\%
            \bar{\mbf{A}}_{21} & \bar{\mbf{A}}_{22}
        \end{bmatrix},
        &
        \mbf{V}^{\trans}\mbf{B}\mbf{V}
        &=
        \begin{bmatrix}
            \bar{\mbf{B}}_{11} & \bar{\mbf{B}}_{12} \\%
            \bar{\mbf{B}}_{21} & \bar{\mbf{B}}_{22}
        \end{bmatrix}.
    \end{align}
    Substituting \cref{eqn:AS_eq_SB_A_bar_B_bar} into \cref{eqn:AS_eq_SB_step_2} results in
    \begin{equation} \label{eqn:AS_eq_SB_step_3}
        \begin{bmatrix}
            \bar{\mbf{A}}_{11}\mbs{\Sigma}_{1} & \mbf{0}\\%
            \bar{\mbf{A}}_{21}\mbs{\Sigma}_{1} & \mbf{0}
        \end{bmatrix}
        =
        \begin{bmatrix}
            \mbs{\Sigma}_{1}\bar{\mbf{B}}_{11} & \mbs{\Sigma}_{1} \bar{\mbf{B}}_{12}\\%
            \mbf{0}                            & \mbf{0}
        \end{bmatrix}.
    \end{equation}
    Equating similar terms in \cref{eqn:AS_eq_SB_step_3} yields \(\bar{\mbf{A}}_{11}\mbs{\Sigma}_{1} = \mbs{\Sigma}_{1}\bar{\mbf{B}}_{11}\), \(\bar{\mbf{A}}_{21}\mbs{\Sigma}_{1} = \mbf{0}\), and \(\mbs{\Sigma}_{1}\bar{\mbf{B}}_{12} = \mbf{0}\). Since \(\mbs{\Sigma}_{1}\) is full rank, it follows that
    \begin{align} \label{eqn:AS_eq_SB_step_5}
        \bar{\mbf{A}}_{11} &= \mbs{\Sigma}_{1}\bar{\mbf{B}}_{11}\mbs{\Sigma}_{1}^{-1},& 
        \bar{\mbf{A}}_{21} &= \mbf{0},&
        \bar{\mbf{B}}_{12} &= \mbf{0}. 
    \end{align}
    Substituting \cref{eqn:AS_eq_SB_step_5} into \cref{eqn:AS_eq_SB_A_bar_B_bar} and rearranging recovers the desired result.
\end{proof}
\vspace{3pt}
\begin{lemma} \label{lemma:AM-QM}
    Suppose that \(u_1, \hdots, u_N\) are real numbers. The arithmetic mean-quadratic mean (AM-QM) inequality, given as part of the mean inequalities in~\cite[Theorem~2.16]{IMO}, states that
    \begin{equation*}
        \mleft( \sum_{i=1}^{N} u_i  \mright)^2 
        \leq 
        N \sum_{i=1}^{N} u_i^2.
    \end{equation*}

    \begin{proof}
        The Cauchy\textendash{}Schwartz inequality for \(\mbf{u} = \mleft[ u_1, \hdots, u_N \mright]^{\trans} \in \mathbb{R}^{N} \) and \(\mbf{v} = \mleft[ 1, \hdots, 1 \mright]^{\trans} \in \mathbb{R}^{N} \) states that
        \begin{equation*}
            \mleft(\sum_{i=1}^{N} u_i\mright)^2 
            =
            \lvert \left\langle \mbf{u}, \mbf{v} \right\rangle \rvert ^2 
            \leq 
            \mleft\| \mbf{u} \mright\|_{2}^{2}
            \mleft\| \mbf{v} \mright\|_{2}^{2}
            =
            N \sum_{i=1}^{N} u_i^2.
        \end{equation*}
    \end{proof}
\end{lemma}
\vspace{3pt}
\begin{lemma}[{\mycite[Lemma 2]{QSR_LMI_OF_LTI}}]\label{eqn:QSR_LMI_Lemma}
    The continuous-time linear time-invariant system described by the state-space realization \(\mleft(\mbf{A}, \mbf{B}, \mbf{C}, \mbf{D}\mright)\) is QSR-dissipative with matrices \(\mbf{Q} = \mbf{Q}^{\trans}\), \(\mbf{S}\), and \(\mbf{R} = \mbf{R}^{\trans}\) if and only if there exists \(\mbf{P} = \mbf{P}^{\trans} \succ 0\) such that
    \begin{equation*}
        \begin{bmatrix}
            \mbf{P}\mbf{A} + \mbf{A}^{\trans} \mbf{P} - \hat{\mbf{Q}} & \mbf{P}\mbf{B} - \hat{\mbf{S}}  \\%
            \mleft( \mbf{P}\mbf{B} - \hat{\mbf{S}} \mright)^{\trans}  & - \hat{\mbf{R}}
        \end{bmatrix}
        \preceq 0,
    \end{equation*}
    where
    \begin{align*}
        \hat{\mbf{Q}} &= \mbf{C}^{\trans} \mbf{Q} \mbf{C},\\%
        \hat{\mbf{S}} &= \mbf{C}^{\trans} \mbf{S} + \mbf{C}^{\trans} \mbf{Q} \mbf{D},\\%
        \hat{\mbf{R}} &= \mbf{D}^{\trans} \mbf{Q} \mbf{D} + \mleft( \mbf{D}^{\trans} \mbf{S} + \mbf{S}^{\trans} \mbf{D} \mright) + \mbf{R}.
    \end{align*}
\end{lemma}
\vspace{5pt}
\begin{theorem}[{\mycite[Theorem 2]{QSR_Stability}}]\label{eqn:QSR_Stability_Theorem}
    Consider two continuous-time systems \(\bm{\mathcal{H}}_1\) and \(\bm{\mathcal{H}}_2\) that are QSR-dissipative with matrices \(\mbf{Q}_1 = \mbf{Q}_{1}^{\trans}\), \(\mbf{S}_1\), \(\mbf{R}_1 = \mbf{R}_{1}^{\trans}\), and \(\mbf{Q}_2 = \mbf{Q}_{2}^{\trans}\), \(\mbf{S}_2\), \(\mbf{R}_2 = \mbf{R}_{2}^{\trans}\), respectively. Provided \(\bm{\mathcal{H}}_1\) is connected in a negative feedback loop with \(\bm{\mathcal{H}}_2\), the feedback system is asymptotically stable if there exists a \(\rho \in \mathbb{R}_{>0}\) such that
    \begin{equation} \label{eqn:QSR_Stability_Theorem_LMI}
        \begin{bmatrix}
            \rho \mbf{Q}_1 + \mbf{R}_2              & -\rho \mbf{S}_1 + \mbf{S}_{2}^{\trans}  \\%
            -\rho \mbf{S}_1^{\trans} + \mbf{S}_{2}  & \rho \mbf{R}_1  + \mbf{Q}_{2}
        \end{bmatrix}
        \prec 0.
    \end{equation}
\end{theorem}
\vspace{5pt}
    \printbibliography
\end{document}